\documentclass[letterpaper]{article} 
\usepackage{aaai24}  
\usepackage{times}  
\usepackage{helvet}  
\usepackage{courier}  
\usepackage[hyphens]{url}  
\usepackage{graphicx} 
\usepackage{natbib}  
\usepackage{caption} 
\usepackage{algorithm}
\usepackage{algorithmic}

\usepackage{newfloat}
\usepackage{listings}
\DeclareCaptionStyle{ruled}{labelfont=normalfont,labelsep=colon,strut=off} 
\floatstyle{ruled}
\newfloat{listing}{tb}{lst}{}
\floatname{listing}{Listing}
\usepackage{amsmath,amssymb}
\usepackage{amsthm}
\usepackage{amsfonts}       
\usepackage{mathtools}
\usepackage{nicefrac}

\usepackage{subfigure}
\usepackage{booktabs}
\usepackage{latexsym}

\usepackage{anyfontsize}
\usepackage{comment}

\newcommand{\commentout}[1]{}

\newcommand{\R}{\mathbb{R}}

\newcommand{\Z}{\mathbb{Z}}

\newcommand{\E}{\mathbb{E}}
\newcommand{\N}{\mathbb{N}}

\newcommand{\calT}{\mathcal{T}}

\newcommand{\calN}{\mathcal{N}}
\newcommand{\calK}{\mathcal{K}}
\newcommand{\calJ}{\mathcal{J}}
\newcommand{\calQ}{\mathcal{Q}}

\newcommand{\prompt}{\texttt{prompt}}

\newcommand{\OPT}{\textsc{OPT}}

\newcommand{\BR}{\mathit{BR}}
\newcommand{\bu}{\mathit{bu}}
\newcommand{\Vis}{\mathit{Vis}}

\newcommand{\Aud}{\mathit{Aud}}

\newcommand{\Act}{\mathit{Act}}

\newcommand{\tS}{\widetilde{S}}
\newcommand{\tA}{\widetilde{A}}

\DeclareMathOperator*{\argmax}{argmax}

\newtheorem{thm}{Theorem}
\newtheorem{theorem}[thm]{Theorem}
\newtheorem{lemma}[thm]{Lemma}

\newtheoremstyle{TheoremNum}%
    {\topsep}{\topsep}
    {\itshape}
    {}
    {\sc}
    {.}
    { }
    {\thmname{#1}\thmnote{ \sc #3}}
\theoremstyle{TheoremNum}
\newtheorem{thmdupl}{Theorem}

\usepackage[textwidth=18mm]{todonotes}

\title{Content Prompting: Modeling Content Provider Dynamics \\ to Improve User Welfare in Recommender Ecosystems}
\author {
    Siddharth Prasad\textsuperscript{\rm 1}\thanks{Work performed as a student researcher at Google Research.},
    Martin Mladenov\textsuperscript{\rm 2},
    Craig Boutilier\textsuperscript{\rm 2}
}
\affiliations {
    \textsuperscript{\rm 1}Carnegie Mellon University\\
    \textsuperscript{\rm 2}Google Research\\
    sprasad2@cs.cmu.edu, mmladenov@google.com, cboutilier@google.com
}

\usepackage{bibentry}

\begin{document}

\maketitle

\begin{abstract}
Users derive value from a recommender system (RS) only to the extent that it is able to surface content (or items) that meet their needs/preferences. While RSs often have a comprehensive view of user preferences across the entire user base, content providers, by contrast, generally have only a \emph{local} view of the preferences of users that have interacted with their content. This limits a provider's ability to offer \emph{new} content to best serve the broader population. In this work, we tackle this \emph{information asymmetry} with \emph{content prompting policies}. A \emph{content prompt} is a hint or suggestion to a provider to make available novel content for which the RS predicts \emph{unmet user demand}. A \emph{prompting policy} is a sequence of such prompts that is responsive to the dynamics of a provider's beliefs, skills and incentives. We aim to determine a \emph{joint} prompting policy that induces a set of providers to make content available that optimizes \emph{user social welfare in equilibrium}, while respecting the incentives of the providers themselves. Our contributions include: (i) an abstract model of the RS ecosystem, including content provider behaviors, that supports such prompting; (ii) the design and theoretical analysis of sequential prompting policies for individual providers; (iii) a mixed integer programming formulation for optimal joint prompting using path planning in content space; and (iv) simple, proof-of-concept experiments illustrating how such policies improve ecosystem health and user welfare. 
\end{abstract}

\section{Introduction}

Recommender systems (RSs) play a critical role in surfacing items (e.g., products, services or content) to their users, especially when the item corpus is vast. Of course, RSs not only create value for their users; they provide a significant service to  \emph{item providers} by helping them identify a market or audience for their offerings. Given the central role RSs play in facilitating the interaction between providers and users, we view the RS as lying at the heart of a complex \emph{ecosystem} comprising all users, all providers, and the RS itself. A healthy recommender ecosystem requires that the item corpus be updated constantly to reflect the ever-changing needs and preferences of its users.

Naturally, one might expect the RS to assist providers in developing \emph{new} content (or products, services, etc.) to meet changing user needs. However, this is rarely done---providers are generally left to their own devices to explore, design and test new offerings. As a result, the ecosystem is often in an \emph{economically inefficient state}, generating less-than-optimal social welfare for its users (and its providers).

Generally, while RSs model user preferences for items \emph{in the corpus}, these models often generalize \emph{out-of-corpus} to some extent, thus revealing user preferences for ``hypothetical'' items the RS has not yet seen. This can be understood as a comprehensive view of \emph{latent consumer demand}. Unfortunately, providers lack this global view, since they generally interact with only a small subset of users. Moreover, the RS has a holistic view of the corpus, and possibly the abilities of providers to source or create new items; this can be interpreted as a deep perspective on \emph{potential supply}, a degree of insight also not generally shared by providers. It is this \emph{information asymmetry} between the RS and the providers that induces economic friction, preventing providers from making optimal design, sourcing or creation decisions regarding new items they might offer to users through the RS.

In this work, we develop a stylized model for understanding this information asymmetry and propose techniques to minimize its impact on (user) social welfare. Specifically, we propose the use of \emph{provider prompting policies} by the RS. These policies suggest novel items a provider might offer, over time, in a way that accounts for: (i) a provider's incentives and beliefs w.r.t.\ audience/market and its skills; (ii) the dynamics of such audience and skill beliefs; and (iii) the potential clash of policy prompts to different providers. Under certain conditions, we show that our prompting policies lead to a socially optimal provider equilibrium in which providers are incentivized to make items available that maximize user welfare (i.e., sum of user utilities for recommended items).

The remainder of the paper is organized as follows. After a brief discussion of related work (Sec.~\ref{sec:related}), we provide a detailed problem formulation of the process dynamics that incorporates: user affinities and social welfare; content provider skills, beliefs, incentives and best responses; and RS matching and prompting policies (Sec.~\ref{sec:model}). In Sec.~\ref{sec:oneprovider}, we detail \emph{single-provider prompting policies} in which the RS suggests a sequence of content points to one provider, assuming others remain fixed. These policies \emph{incentivize} the provider to follow a content path that reaches a socially optimal equilibrium (w.r.t.\ user welfare) in polynomial time. We extend this model in Sec.~\ref{sec:joint} to allow for the coordinated, \emph{joint prompting} of all providers,
and develop a mixed integer programming (MIP) formulation to solve the induced \emph{multi-agent path planning problem}. 
In Sec.~\ref{sec:experiments}, we evaluate our procedures on random problem instances to show the significant improvement in social welfare that can be created by prompting providers. While somewhat stylized, our model provides important first steps to a more \emph{holistic, incentive-aware design of recommender ecosystems}.

\section{Related Work}
\label{sec:related}

RSs typically rely on some form of 
prediction of a user's interests or intent, based on past usage or ratings data
\citep{grouplens:cacm97,jacobson2016music,covington:recsys16}. We do not commit to a specific method for generating such models, but our approach can be understood in terms of latent-space user and item embeddings of the type generated by collaborative filtering (CF) methods (e.g., matrix factorization \cite{salakhutdinov-mnih:nips07} or neural CF \cite{he_etal:www17,beutel_etal:wsdm18}) to estimate user affinity to specific items/content.

Recently there has been a growing appreciation for the multiagent interactions between users, or user groups, and content providers in an RS~\cite{burke_multistakeholder:RecSysHandbook2021}. In this work, we focus on the incentives and behaviors of providers and assume user interests are fixed. As such, of some relevance is work that considers ``fairness of exposure'' to the items offered by different providers, and modifications of the RS policy to increase such fairness~\cite{SinghJoachims:kdd18,Biega,diaz:sigir22,heuss:sigir22}.

More directly connected is work that explores the behavior of providers in response to RS policies. Of special relevance is a recent line of work investigating the incentives that providers have to change the content they offer, and game-theoretic equilibria in such settings.
\citet{benporat_etal:nips18} develop a game-theoretic
model of RSs whose providers act strategically---by making available
or withholding content---to maximize their user engagement in an RS.
\citet{benporat_etal:aaai19} draw a direct
connection between facility
location (or Hotelling) games and RSs with strategic providers. 
Both \citet{hron_creator_incentives:arxiv22} and \citet{jagadeesan_supply:arxiv22} study multi-dimensional content models, showing under what conditions Nash equilibria exist, how providers ``position'' themselves in content space, and the effect this can have on (say) generalization vs.\ specialization of content production. While this work is similar to ours in its study of adaptive provider behavior, it differs by assuming that providers have full knowledge of supply and demand, something that is far from true in most RSs. By contrast, we focus on the \emph{information asymmetry between providers and the RS}, and the active role the RS can play to reduce it. Our approach allows self-interested providers to make more informed content decisions that induce equilibria that better serve both users and providers. We also adopt a more nuanced model of provider skill and beliefs.

We focus on how the RS can intervene to shape the desired equilibrium.
\citet{benporat_etal:aaai19} (see above) examine Nash equilibria in the 1-D case under different RS matching policies while~\citet{ben2020content} extend and generalize this analysis further. \citet{mladenov_etal:icml20} investigate (non-strategic) provider behaviors driven by the engagement they obtain, and RS policies that optimize long-term user welfare via matchings that anticipate the equilibria of the induced dynamical system. Neither model addresses information asymmetry or active provider intervention.

\section{Problem Formulation}
\label{sec:model}

We begin with a formulation that adopts a somewhat stylized model of a content RS, based on user and item embeddings. We also adopt a simplified model of content provider \emph{skills} and \emph{beliefs}, their dynamics, their decisions, and the RS knowledge of these elements. While simpler than a real RS ecosystem, this model contains the essential elements required to reason effectively about provider-prompting policies. Our multi-stage model of the RS ecosystem uses a single-stage model based on that of \citet{mladenov_etal:icml20}, though we develop a very different dynamics model. 

\subsection{Providers, Users and One-stage Recommendations}

We assume an RS that matches \emph{users} to the content made available by \emph{content providers}. 
Our process is organized into $T$ stages. We use $T$ for both finite and infinite-horizons, taking $\lim_{T\to\infty}$ in the latter case.)
At the beginning of each stage $t\leq T$, each provider determines the content they will make available. During the stage, as users arrive, the RS matches them to providers (and their content) using some fixed policy.

\vskip 2mm
\noindent
\textbf{Providers, Content Points and Skill}: We assume a finite set of \emph{providers} $\calK = \{1,\ldots, K\}$. At stage $t$, each provider determines the content it will make available for recommendation, from a finite set of \emph{content points} $\calJ = \{1,\ldots, J\} \subset \R^d$, where each point is an embedding of that content.\footnote{This can be generalized to a continuous set of points in some embedding space. However, our user utility and provider reward models ensure a provider only chooses from a finite set of points in equilibrium. This contrasts with approaches that use simpler provider rewards (e.g., those that simply count matched users \cite{jagadeesan_supply:arxiv22}, which can induce a continuum of equally good equilibria).}
Let $\ell^t_k \in \calJ$ be the content generated by provider $k$ at stage $t$, which we call $k$'s \emph{location} at $t$.  The \emph{location vector} $L^t = \langle \ell^t_1, \ldots, \ell^t_K\rangle$ reflects the content decisions of all providers. 

Some providers are more adept at producing certain types of content than others (e.g., due to specific talent, interests, facilities, etc.). Let $s_{k,j}\in [0,1]$, be $k$'s (true) \emph{skill} w.r.t.\ to point $j$, reflecting this aptitude. Let $S$ be the (true) \emph{provider skill matrix}, and $S_k$ (the $k$th row of $S$) the \emph{skill vector} for provider $k$. We treat skill as fixed. Skill will exhibit generalization across ``similar'' content points, but we make no such assumption here.

\vskip 2mm
\noindent
\textbf{Users, Affinity, Utility and Reward}: We assume a finite set $\calQ \subset \R^d$ of users (or queries). At stage $t$, a set of users/queries is drawn from (known or estimated) distribution
$P^t(\calQ)$.
When $|\calQ|$ is small, we interpret each $q\in\calQ$ as representative of a \emph{user type}. For ease of exposition, we take $P^t$ to be uniform.
Each user $i\in\calQ$ has an inherent \emph{affinity} for content $j\in\calJ$ given by a non-negative, bounded \emph{affinity function} $\sigma(i,j)$, e.g., dot product or cosine similarity
if we use some \emph{latent space} $X\subseteq\R^d$ to embed users and content,
as in matrix factorization or neural CF. Together with affinity, the skill of the provider dictates user utility:  
if $j$ is offered by provider $k$ and recommended to $i$, $i$'s \emph{utility} for $j$ is given by \emph{utility function} $f(\sigma(i,j), s_{k,j})$. 
This might be as simple as the product $\sigma(i,j)\cdot s_{k,j}$.
 
We assume a provider's \emph{reward} for that recommendation is equal to the user's utility. Increasingly, RSs and content providers focus on user satisfaction  beyond simple engagement metrics \cite{diaz_et_al:recsys18,goodrow_youtube_blog_2021} and long-term utility; we use the term ``utility,'' rather than engagement, to emphasize this.  We assume the RS aims to maximize total user utility. Equating user utility with provider reward also aligns RS, user, and provider incentives; but if providers optimize for different forms of engagement, that can be incorporated below. 

\vskip 2mm
\noindent
\textbf{Matching Policies and Value}: When user $q$ arrives, the RS recommends the content of some provider $k$. A \emph{(stochastic) matching} $\mu: \calQ \rightarrow \Delta(\calK)$ associates each $q\in \calQ$ with a distribution over providers.\footnote{We adopt this model for simplicity. Generally, users issue different queries, or request recommendations in different contexts, under which the best match differs, a detail handled by all modern RSs. Likewise, providers offer multiple content items, and RSs match to specific provider content. Our formulation applies \emph{mutatis mutandis} if we match \emph{queries}
to provider \emph{content} rather than users to providers.} 
We write $\mu(q,k)$ to denote the match probability.
The \emph{value} of $\mu$ given a location vector $L^t$ is its expected user utility: 
$$V(\mu, L^t) = \sum_{q\in\calQ} \sum_{k\in\calK}
   \mu(q,k) f(\sigma(q,\ell^t_k), s_{k,\ell^t_k}).$$
The \emph{natural matching} $\mu^\ast_{L^t}$ which maximizes this value would be optimal if $L^t$ were stationary, but given the dynamics below, we avoid the term ``optimal'' to describe this policy. Implementing $\mu^{\ast}_{L^t}$ requires complete knowledge of user and content embeddings and provider skills, which we assume the RS has.

\subsection{Provider Beliefs and Decisions, and RS Prompts}

We now describe the process by which the $T$ stages unfold, including how providers update their beliefs about their skills and audience, and how these beliefs influence their content decisions.

\vskip 1mm
\noindent
\textbf{Skill and Audience Beliefs}: Providers base their decisions about which content to make available 
at stage $t$ 
by estimating their own audience's utility. This requires estimating both their own skill and the (affinity-weighted) audience they expect to be generated for them by the RS for any
content point
$j\in\calJ$ they might offer. Each provider $k$ has a \emph{skill belief vector} $\widetilde{S}^t_k$, where entry $\widetilde{s}^t_{k,j}\in [0,1]$ denotes $k$'s \emph{estimated skill or skill belief} w.r.t.\ point $j$. Let $\widetilde{S}^t$ be the matrix of such skill vectors ($k$'s beliefs are over its own skill, not that of other providers $i\neq k$).

While true skill $S$ is fixed, estimates $\widetilde{S}^t$ change over time as, say, the provider gains experience with new content.
We assume a skill belief update function of the form $\widetilde{S}_k^{t+1} = \bu_S(L^t_k, \widetilde{S}_k^{t}, E^t_k)$, i.e., $k$'s skill belief at stage $t+1$ depends on its prior belief $\widetilde{S}_k^{t}$ and the utility $E^t_k$ garnered at its location $L^t_k = j$. For simplicity, we sometimes assume a simple deterministic (and non-generalizing) skill update $\bu_{S,D}$: if $\ell^t_k = j$, $k$'s belief for point $j$ collapses to its true skill $s_{k,j}$ for all $t'> t$, with no other $j'$ being updated at time $t+1$. This is easily generalized, though the optimization below requires (an estimate of) an update model.

Each provider $k$ also has an \emph{audience (affinity) belief} $\widetilde{a}^t_{k,j}\le Q$ that measures its estimate of the total audience affinity it expects to attain if it offers content at point $j$. Define vector $\widetilde{A}^t_k$ and matrix $\widetilde{A}^t$ in the obvious way. These estimates also vary with time.
As with skills, we assume an audience update function $\widetilde{A}_k^{t+1} = \bu_A(L^t_k, \widetilde{A}_k^{t}, E^t_k)$, and sometimes assume a simple deterministic model $\bu_{A,D}$: if $\ell^t_k = j$, $k$'s audience belief for point $j$ collapses to its realized audience $A_{k,j}$ at time $t+1$,
so $\widetilde{A}^{t+1}_{k,j} = A_{k,j}$. Unlike skill beliefs, these can change with each new experience at point $j$. Together with deterministic skill update, this means a provider can determine the (expected) total user affinity---assuming the \emph{number} of users is observable while affinity is estimated---to which is was matched from the utility signal under simple, say, linear utility models (e.g., if user utility is the product of provider skill and user affinity). This too is easily generalized.
%
One exception to this form of audience belief update is if the RS provides a \emph{prompt}, defined next.

\vskip 1mm
\noindent
\textbf{Prompts}: To encourage providers to generate content that increases user utility, the RS uses \emph{prompts}, that is, suggestions to providers to make content available at certain points. To incentivize such content, the RS can (temporarily) commit some amount of audience affinity to a provider. Formally, an RS \emph{prompt} $\nu^{t-1}_k = (j, E)$ of provider $k$ at time $t-1$  consists of: (i) a suggested content point/location $j = \ell^t_k$ for $k$ to produce at stage $t$; and (ii) a commitment to a (minimum) level of audience utility $E = E^t_k$ at stage $t$ if $k$ produces $j$.

The effect on $k$'s audience belief depends on their level of \emph{trust} in the RS. Let $\lambda^t_k \in [0,1]$ denote $k$'s current trust; then $\widetilde{A}^{t}_j  \leftarrow \bu_P(\widetilde{A}^t_j,E,\lambda^t)$, where $\bu_P$ is an update function that determines $k$'s \emph{prompted belief} used to make its next content decision (see below). This prompted belief is tentative, as it will be updated given the \emph{realized audience} at stage $t+1$ (via $\bu_A$). We consider two example updates. The first is \emph{incremental-trust belief update}, where:
(i) $\bu_{P,I}(\widetilde{A}^t_j,E,\lambda^t) = (1-\lambda^t)\cdot \widetilde{A}^t_j + \lambda^t\cdot E$; and (ii) the trust parameter $\lambda^{t+1} = \tau(\lambda^t, E, E^t_r)$ is updated given $k$'s realized utility using a \emph{trust update function} $\tau$.  The second is \emph{full-trust belief update}, where  $\bu_{P,F}(\widetilde{A}^t_j,E,\lambda^t) = E$.

\vskip 1mm
\noindent
\textbf{Process Dynamics}: RS dynamics evolve as follows. The state $s^{t} = \langle L^{t}, \widetilde{S}^{t}, \widetilde{A}^{t} \rangle$ at stage $t$ consists of: location vector $L^t$, skill belief matrix $\widetilde{S}^t$, and audience belief matrix $\widetilde{A}^t$. The location $\ell^t_k$ chosen by provider $k$ is consistent with their beliefs (see below). The RS generates a matching $\mu^t$, the realization of which induces value for users, and utility for providers. Based on their realized utility, each provider updates its (skill and audience) beliefs---thought of as a \emph{half state} $s^{t+\frac{1}{2}} = \langle L^{t}, \widetilde{S}^{t+\frac{1}{2}}, \widetilde{A}^{t+\frac{1}{2}} \rangle$ (locations do not yet change). The RS may then \emph{prompt} providers by providing information $\nu^t$ about the audience they will receive if they offer content at a specific point, which can induce a further update in their audience belief. Prompts influence only audience beliefs, not skill beliefs, so $\widetilde{S}^{t+1} = \widetilde{S}^{t+\frac{1}{2}}$. Given updated beliefs $\widetilde{S}_k^{t+1}$ and  $\widetilde{A}_k^{t+1}$, each $k$ selects their next location $\ell^{t+1}_k$, determining state $s^{t+1}$. The process, schematically represented below, then repeats
$$\textstyle \langle L^t, \widetilde{A}^t, \widetilde{S}^t\rangle\xrightarrow{\mu} \langle L^t, \widetilde{A}^{t+\frac{1}{2}}, \widetilde{S}^{t+\frac{1}{2}}\rangle \xrightarrow{\nu}\langle L^{t+1}, \widetilde{A}^{t+1}, \widetilde{S}^{t+1}\rangle$$

\noindent
\textbf{Provider Best Responses}:\hspace{1mm} To model provider choices of location at each stage,
we assume providers are \emph{myopic utility maximizers} w.r.t.\ their own beliefs, but are not strategic in the ``dynamic'' sense (that is, they do not reason about how their actions might impact the RS policy or the behavior of other providers). Thus, given beliefs $\widetilde{S}^t_{k}, \widetilde{A}^t_{k}$ at stage $t$, $k$ chooses the location $\ell^t_k$ corresponding to their \emph{best response}, i.e., the content point
$\ell^t_k = \BR(\widetilde{S}^t_{k}, \widetilde{A}^t_{k}) =  \arg\max_{j\in\calJ} \widetilde{s}^t_{k,j} \widetilde{a}^t_{k,j}$
for which it predicts the greatest utility.
We call state $s^t = \langle L^t, \widetilde{S}^t, \widetilde{A}^t \rangle$ \emph{rationalizable} if $\ell^t_k = \BR(\widetilde{S}^t_{k}, \widetilde{A}^t_{k})$ for all providers $k$. We assume that provider location choice satisfies rationalizability in what follows.

\vspace*{2mm}
\noindent
\textbf{Stable Matchings and States}: A matching is \emph{(myopically) stable} if---with no RS prompts---it gives no incentive for a location change. Let $s^t = \langle L^t, \widetilde{S}^t, \widetilde{A}^t \rangle$, let 
$$E_k^t = \sum_{i\in\calQ} \mu^t(i,k) f(\sigma(i,\ell^t_k), s_{k,\ell^t_k})$$
be $k$'s realized utility under $\mu^t$,
and let $$\widetilde{S}_k^{t+1} = \bu_S(L^t_k, \widetilde{S}_k^{t}, E^t_k)$$ be $k$'s updated skill belief.
Without prompts, $\widetilde{A}_k^{t+1} = \widetilde{A}_k^{t+\frac{1}{2}}$.
We say that $\mu^t$ is \emph{myopically stable} w.r.t.\ $s^t$ if, for all $k\in\calK$:
      $$\widetilde{a}^{t+1}_{k,\ell^t_k} \widetilde{s}^{t+1}_{k,\ell^t_k} \geq \widetilde{a}^{t+1}_{k,j} \widetilde{s}^{t+1}_{k,j}, \; \forall j \neq \ell^t_k,$$
i.e., $k$'s current location $\ell^t_k$ remains a best response after it experiences the utility induced by $\mu^t$.

While myopic stability is a desirable equilibrium property, it does not ensure stability/equilibrium of the dynamical system. A matching/state may be myopically stable simply due to the slowness of the process by which providers update their beliefs. A true equilibrium notion requires that a stable state persists indefinitely: we say state $s^t$ is \emph{(non-myopically) stable} if there is a sequence of matchings $\mu^t, \mu^{t+1}, \ldots$ such that for all providers $k$ and stages $t' > t$,
$$\widetilde{a}^{t'}_{k,\ell^t_k} \widetilde{s}^{t'}_{k,\ell^t_k} \geq \widetilde{a}^{t'}_{k,j} \widetilde{s}^{t'}_{k,j}, \; \forall j \neq \ell^t_k.$$
A matching $\mu$ is \emph{(non-myopically) stable} w.r.t.\ $s^t$ if the matching sequence $\mu^{t'} = \mu, \forall t' \geq t$ renders $s^t$ stable in the sense above. Once a desirable system state is reached, ideally it will be stable w.r.t.\ a \emph{single} matching.


\vspace*{2mm}
\noindent
\textbf{Overall Objective}: Our overall objective is the maximization of user social welfare: this may be $V(\mu, L^T)$ at stage $T$ in the finite-horizon case, or long-term expected average reward $\lim_{T\to\infty}\E[\frac{1}{T}\sum_{t=0}^{T-1} V(\mu^t, L^t)]$. 
If provider content decisions are fixed/cannot be influenced by the RS, this simply requires the application of the natural matching $\mu^\ast_{L^T}$ relative to $L^T$. At the other extreme, if the RS could persuade providers to move to arbitrary locations at will, the problem is akin to a facility location or $k$-medians problem \citep{whelan2015understanding} where users are clients whose (inverse) affinities reflect travel costs and provider skill captures service costs. 
The reality of course is different---providers are independent decision makers whose content decisions accord with their beliefs and incentives.

The aim of a \emph{prompting policy} is to break the information asymmetry described above to allow providers to better calibrate their audience and skill beliefs so that their decisions lead to a (close to) welfare-optimizing matching. Moreover, since providers make their own content decisions, we want the final matching to be in \emph{equilibrium}, that is, be stable given providers' incentives.
The RS also uses its \emph{matching policy} by (perhaps temporarily) creating specific audience/utility levels that incentivize providers to make suitable moves. Taken together,
an RS policy has two parts: a \emph{matching policy} that associates a matching $\mu^t$  with each state $s^t = \langle L^t, \widetilde{S}^t, \widetilde{A}^t \rangle$; and a \emph{prompting policy} that, for each half state $s^{t+\frac{1}{2}} = \langle L^t, \widetilde{S}^{t+\frac{1}{2}}, \widetilde{A}^{t+\frac{1}{2}} \rangle$, generates a prompt $\nu^t$.

\section{Single-Provider Prompting Policies}
\label{sec:oneprovider}

A cornerstone of our approach and analysis is the \emph{single-provider prompting policy}. Here, we assume that the content decisions of all providers except one, $k\in\calK$, are fixed, and construct a prompting policy that ensures $k$ moves to the location $j\in\calJ$ that maximizes user welfare given the fixed locations of the other providers. Specifically, assume a fixed location vector $L_{-k}$ that specifies $\ell_{k'}$ for all $k'\neq k$, where $\ell^t_{k'} = \ell_{k'}$ for all $t\leq T$;
and an initial state
$s^0 = \langle L^0, \widetilde{S}^0, \widetilde{A}^0 \rangle$.
Our goal is to design a joint matching/prompting policy that induces $k$ to move to a target $j^\ast_k \in \calJ$, inducing location vector $L^\prompt = L_{-k}\circ j_k^\ast$, that is optimal w.r.t.\ long-term average reward:
$$\lim_{T\to\infty}\E\left[\frac{1}{T}\sum_{t=0}^{T-1} V(\mu^\ast_{L^\prompt},L^\prompt)\right].$$
Furthermore, the policy should induce an \emph{equilibrium}, i.e., a stable state where $k$ offers $j^\ast$ without additional prompts. Proofs of all results in this section along with additional details of our formulation are provided in Appendix~\ref{sec:app_oneprovider}.

\subsection{Equilibrium with No Prompting}

We illustrate the value of prompting policies with a simple example to show the potential loss in total user utility/welfare
that accrues without prompting.
Assume initial state $s^0$, where $k$'s beliefs are $\widetilde{S}_k^0$ and $\widetilde{A}_k^0$, a matching policy $\mu_L$ that is fixed for any location vector $L$ (e.g., the natural matching), and let $E_{k,j}$ be $k$'s utility under $\mu_L$ if $\ell_k = j$.

Let $O\subseteq\calJ$ be $k$'s \emph{undominated overestimates}---these are content points $j$ for which $k$ initially overestimates expected utility, but for which no other point's \emph{true} and \emph{estimated} values exceed $j$'s true utility: 
\begin{multline*}
O = \{j : \widetilde{s}^0_{k,j}\widetilde{a}^0_{k,j}  > E_{k,j}, \nexists j' \\
   \quad\quad\text{ s.t. } ( \widetilde{s}^0_{k,j'}\widetilde{a}^0_{k,j'}  > E_{k,j} \text{ and } E_{k,j'} > E_{k,j} )\}. 
\end{multline*}
Under reasonable (e.g., monotonically converging) belief updates, and best response behavior, $k$ will eventually try all undominated, overestimated content points. For instance, under immediately collapsing skill beliefs, the system reaches equilibrium in exactly $|O|$ steps: for any undominated $j$, any other $j'$ whose estimated utility is greater than $j$'s estimate is tried before $j$, but $k$'s beliefs immediately collapse to give the true value of $j'$, which by assumption is less than that of $j$. Thus each undominated $j$ will be tried once. Critically, \emph{no dominated point} $j''$ will ever be tried, since at least one undominated point has a \emph{true} utility greater than $k$'s estimate for $j''$.

This lack of exploration can lead to an \emph{arbitrarily suboptimal} equilibrium. For instance, suppose there are only  two points such that $\widetilde{s}^0_{k,2}\widetilde{a}^0_{k,2}  >  E_{k,1} > E_{k,2} > \widetilde{s}^0_{k,1}\widetilde{a}^0_{k,1}$. Then $k$ only offers the suboptimal point $2$, whose true value is greater than $k$'s estimate for point $1$, despite the fact that $1$'s true value is greater than $2$'s. An RS prompt can reveal to $k$ the true utility at point $2$, thus incentivizing $k$ to offer $2$ (and realize its true value). This issue is further exacerbated by the fact that $k$ will never visit {\em any dominated point}, as detailed above.

The provider thus faces an exploration problem: indeed, natural schemes like ``optimism in the face of uncertainty'' would ensure $k$ tries points such as point 1. However, the cost of exploration (e.g., creating new content), the lack of direct control (e.g., predicting audience), and inherent risk aversion may prevent adequate \emph{self-exploration} by providers. Prompting reduces this risk by providing additional certainty, through promised utility, and incentivizing behavior that allows the provider to update their audience/skill beliefs where they might otherwise not. That said, we note that in some cases, provision of certain points may not be incentivizable by the RS, e.g., when $k$ underestimates it skill $s_{k,j}$ so drastically that no promised audience utility can incentivize $k$ to move to $j$, even with full trust.

\subsection{Prompting under Non-generalizing Belief Updates}

Let $\mu^t$ be a stable matching w.r.t\ the current state $s^t$, with $\ell^t_k = j$ and $E^k_j$ being $k$'s expected utility/engagement. The willingness of $k$ to move to the RS target $j^\ast$ depends on its trust in the RS prompts. We refer to Appendix~\ref{sec:app_oneprovider_nongeneralizing} for details (our theorems require specific assumptions on how provider trust is updated), but at a high-level, our prompting policies comprises two phases: (1) the RS takes steps to increase $k$'s level of trust in the RS by promising $k$ its expected audience (given $\mu^t$) at its \emph{current} point $\ell_k^t$---and continuing to match users using $\mu^t$---until $k$'s trust reaches a sufficient level; (2) the RS then prompts $k$ with $(j^\ast, \phi)$, where $\phi$ is the minimum level of audience affinity given $k$'s skill beliefs needed to induce the move; and matches using the optimal target policy $\mu^\ast_{L^\prompt}$. More precisely, our policy, parameterized by a trust threshold $\Lambda$ and a belief convergence rate $T(\varepsilon,\delta)$, is:
\begin{enumerate}
    \item While $\lambda_k < \Lambda$, repeat match-prompt pair $$(\mu,\nu) = \big(\mu^t, (\ell_k, \textstyle\sum_{q\in\calQ}\mu^t(q,k)\sigma(q,\ell^t_k))\big).$$
    \item Issue match-prompt pair $$(\mu^t, (j^*,\phi(j^*; \Lambda, \tS_k, \tA_k))).$$
    \begin{enumerate}
        \item If Step 2 has been done $ T(\varepsilon,\delta/2)$ times, terminate.
        \item Else if $\lambda_k\ge \Lambda$ go to Step 2.
        \item Else if $\lambda_k < \Lambda$ go to Step 1.
    \end{enumerate}
\end{enumerate}

If $\phi(j^*; \Lambda, \tS_k, \tA_k)\le\sum_{q\in\calQ}\mu^t(q,k)\cdot\sigma(q, j^*)$ then $\mu^t$ itself (stochastically) delivers the promised audience. If not, the provider trust might take a hit, but is rebuilt in step (c). First, assume the special case where $\mu^t$ is deterministic and $k$'s beliefs collapse immediately. Then, we may run the above policy with $\Lambda = 1$ and ignore steps (a)-(c). Provider $k$'s trust updates are governed by a learning rate $\eta$.

\begin{theorem}\label{thm:reachstatedeterm}
Let $s^t = \langle L^t, \tA^t, \tS^t\rangle$ be rationalizable and $\mu^t$ be a non-myopically stable matching w.r.t.\ $s^t$. After $\zeta=1/\eta$ iterations of the above policy, the resulting state $s^{t+\zeta}$ remains rationalizable and $\mu^t$ non-myopically stable w.r.t.\ $s^{t+\zeta}$. Furthermore, it maximizes the minimum trust, $\min_{t}\lambda^t$, among all policies that reach $s^{t+\zeta}$.
\end{theorem}

More generally, we prove the complexity of the above policy under the assumption of a sample complexity measure $T(\varepsilon,\delta)$ that determines the number of times a content item must be visited for $k$'s predicted utility to be nearly perfect.

\begin{theorem}\label{thm:reachstategeneral}
Let $|\calQ|\ge \Omega(1/(1-\Lambda)^2)$, let $s^t = \langle L^t, \tA^t, \tS^t\rangle$ be rationalizable, and $\mu^t$ be non-myopically stable w.r.t.\ $s^t$. If we run the above policy with $\varepsilon < \frac{1}{2}\min_{j, j'}|\E[E_{k,j}] -\E[E_{k,j'}]|$, the following hold w.h.p.: (1) the resulting state $s^{t+\zeta}$ has $\ell_k = j^*$, and remains rationalizable, while $\mu$ remains non-myopically stable w.r.t.\ $s^{t+\zeta}$; and (2) the policy terminates in $\zeta=O(T(\varepsilon,\delta)^3\Lambda^2/(\eta^2\delta)^2)$ rounds.
\end{theorem}

As $\Lambda$ varies, our policy trades off (1) quickly reaching equilibrium and (2) keeping trust high. See Appendix~\ref{sec:app_oneprovider_nongeneralizing} for further details and discussion.

\subsection{Prompting when Beliefs Generalize across Content}

The policy above will succeed only if $k$'s skill belief at $j^\ast$ is such that the promised audience will induce it to move. The fact that provider and cumulative user utility are aligned ensure that $\mu^\ast_{L^\prompt}$ will, in fact, satisfy $k$ if it can be convinced to move. If $k$'s skill beliefs for any $j'$ are only updated when it generates content at $j'$, then this limits the set of reachable points, and if $j^\ast$ is reachable, it can be reached with a single prompt. By contrast, suppose $k$'s skill beliefs about $j'$ can be influenced by its experience at a \emph{different} $j''$; e.g., if producing successful snorkelling content increases $k$'s belief it can produce scuba content. Then the RS faces a \emph{path planning problem}: determining a sequence of points $j = j_0, j_1, \cdots j_{n-1}, j_n = j^\ast$ such that each $j_i$ is \emph{incentivizable} (i.e., can be prompted using the two-phase policy above) given $k$'s skill beliefs given that it has generated content at all $j_{\leq i}$. We assume only the following about generalizing belief updates: (1) beliefs at unvisited points never become less accurate, and (2) belief updates for unvisited points are independent of the order in which other points are visited. Full details are provided in Appendix~\ref{sec:app_oneprovider_generalizing}.

A trivial policy that eventually reaches the optimal incentivizable target $j^*$ prompts $k$ to visit any reachable point $j'$ the moment $j'$ is incentivizable. Of course, this may induce visits to unnecessary points. In Appendix~\ref{sec:app_oneprovider_generalizing}, we formulate a MIP to find the \emph{shortest promptable content path}. While this problem is NP-hard, fortunately, there is a simple, efficient greedy algorithm: at each step, if $j^*$ is incentivizable, prompt with $j^*$ and terminate; otherwise prompt with the (currently) incentivizable point $j$ that maximizes the increase in accuracy in $k$'s predicted reward at $j^*$ if $j$ is visited.

\begin{theorem}\label{thm:greedy}[Informal]
The greedy algorithm runs in polynomial time and returns a content path at most a constant factor longer than the shortest content path.
\end{theorem}
See Appendix~\ref{sec:app_oneprovider_generalizing} for a formal statement and proof.

\section{Joint Prompting Policies}
\label{sec:joint}
The problem of simultaneously prompting a set of providers, that is, ensuring that the \emph{joint} location (content) of all providers serves users effectively,  presents a non-trivial escalation in complexity. However, the main insight above---viewing provider evolution as \emph{paths through content space}---suggests treating audience availability for \emph{multiple providers} as the key requirement and challenge in designing \emph{joint prompting policies}.
This introduces an additional planning component to the problem, requiring a joint prompting policy to \emph{schedule prompts} in a way that ensures enough available audience to fulfil the promises needed to incentivize all providers to move to their optimal locations.

This unfortunately rules out the greedy application of the single provider policies above, which can be suboptimal. Consider a simple counterexample with: two providers $a, b$; two points $x,y$; current locations $\ell^t_a = x, \ell^t_b = y$; skills such that $s_{a,y} > s_{a,x}$, $s_{a,y} > s_{b,y}$,  $s_{b,x} > s_{b,y}$ and $s_{b,x} > s_{a,x}$; and user affinities such that the optimal matching $\mu^\ast$ has equal total affinity at points $a$ and $b$ whenever providers at those points have skills at least (resp.) $s_{a,x}, s_{b,y}$. In this situation, a Pareto improving move is to swap the locations of $a$ and $b$, since each has greater skill at the opposite location than at their current location and than their peer. However, any policy that moves one provider at a time must have both providers at the same location for at least one period,  imposing significant cost on user utility.

\commentout{
While 
we have thus far dealt with prompting a single provider to maximize user social welfare, our ultimate goal is ensuring that the \emph{joint} location (content) of all providers serves users effectively. We now turn to this question.

\subsection{Exploiting Single-Provider Prompting Policies}
\label{ref:single_provider}

In principle, single-provider prompting policies can be sequenced to achieve good joint states. Specifically, we can repeat the steps: (i) select a provider $k$ to prompt, (ii) compute an optimal target location $j^\ast_k$ for $k$ \emph{assuming} that other providers remain fixed, (iii) implement a prompting policy for $j^\ast_k$. This procedure iterates until equilibrium where no single provider movement improves social welfare. While prompting a single provider at a time may work in practice, it can be suboptimal. A simple counterexample involves a system with two providers who would be better off swapping locations (see Appendix~\ref{sec:app_mips}). Any policy that moves one provider at a time must have both providers at the same location for at least one period,  imposing significant cost on user utility. 
}

\subsection{Mixed Integer Programming Formulation}
\label{ref:sec_mips}

The complications above mean that joint prompting policies must plan provider paths under ``sufficient audience'' constraints. We develop a MIP formulation to solve this problem over a finite horizon $T$.\footnote{We briefly discuss alternative approaches below.} \commentout{In the first, we directly model possible locations of each provider at each stage. In the second, we implicitly consider all feasible \emph{location-vector paths}, and manage the combinatorial set of variables using {\em column generation (CG)} to avoid explicit enumeration.}

In our formulation, provider behavior is represented by decision variables $\Act_{k,j}^t\in\{0,1\}$ indicating whether provider $k$ offers content $j$ at stage $t$, with the requirement that each $k$ be active at exactly one point $j$ at each $t< T$.
Decisions by the RS are captured by three sets of optimization variables: the matching policy $\pi_{q, k}^t\in [0,1]$; the prompting policy $\nu_{k, j}^t\in\{0,1\}$ indicating whether the RS prompts $k$ to offer content $j$ at stage $t$; and the commitments, $C^t_{k, j}\in\R^+$, or audience promised for any such prompt.

With these variables, we can 
express relevant key quantities:
\begin{itemize}
    \item Provider utility: $$E_k^t = \sum_{q\in\calQ}\pi^t_{q,k}\sum_{j\in\calJ} \Act_{k,j}^t s^*_{k, j}\sigma(q, j);$$
    \item Provider audience: $$\Aud_k^t = \sum_{q\in\calQ}\pi^t_{q,k}\sum_{j\in\calJ} \Act_{k,j}^t\sigma(q,j);$$ 
    \item Induced skill beliefs: We assume skill belief collapses to true skill the moment $j$ is visited, hence:
    $$\widetilde{s}^t_{k,j} = \sum_{j\in\calJ}s^0_{k,j}(1-\Vis_{k,j}^{t-1}) + s^*_{k,j}\Vis_{k,j}^{t-1}$$ where $\Vis_{k,j}^t$ indicates whether $k$ has visited $j$ at or before stage $t$.
\end{itemize}
We outline only key MIP constraints that govern the RS dynamics here; the full MIP is detailed in Appendix~\ref{sec:app_mips}. Indices are universally quantified. We use $\lesseqgtr$ and $\pm$ to denote two constraints (in the obvious way).
$M$ in the ``big-$M$'' constraints is an upper bound on audience utility.
\begin{align}
    &C_{k, j}^t\le\nu_{k,j}^t M
            \label{eq:con1a} \tag{Cons.~1(a)} \\     \nonumber
    &\Aud_k^t\ge C_{k,j}^{t-1} \Act_{k,j}^t
            \label{eq:con1b} \tag{Cons.~1(b)}  \\     \nonumber
    &\widetilde{s}^t_{k,j} \widetilde{a}^t_{k,j}\ge \widetilde{s}^t_{k,j'} \widetilde{a}^t_{k,j'} - (1 - \Act_{k,j}^t)M
            \label{eq:con2} \tag{Cons.~2} \\     \nonumber
  &\widetilde{a}^t_{k,j}\lesseqgtr \Aud_k^{t-1} \pm (1-\Act_{k,j}^{t-1})M \pm \nu_{k,j}^{t-1}M
            \label{eq:con3} \tag{Cons.~3} \\      \nonumber
  &\widetilde{a}^t_{k,j}\lesseqgtr \widetilde{a}^{t-1}_{k,j} \pm \Act_{k,j}^{t-1}M \pm \nu_{k,j}^{t-1}M
            \label{eq:con4} \tag{Cons.~4} \\     \nonumber
  &\widetilde{a}^t_{k,j}\lesseqgtr C_{k,j}^{t-1} \pm (1-\nu_{k,j}^{t-1})M
            \label{eq:con5} \tag{Cons.~5}     \nonumber
\end{align}
\ref{eq:con1a} and~\ref{eq:con1b} are prompting constraints: the first limits the promised audience if a prompt is given, and ensures the promise is zero if no prompt is given; the second ensures that if $k$  ``accepted'' a prompt (i.e., moved to the prompted location), the matching delivers the promised audience.
\ref{eq:con2} ensures providers only offer content that are best responses.
The final constraints control audience belief updates: \ref{eq:con3}  when $k$ is at $j$ but was not prompted (belief updated to realized audience); \ref{eq:con4} at points that are inactive and unprompted (beliefs persist);
 and \ref{eq:con5} handles prompts (belief updated to promised audience).




Our objective is to maximize time-averaged user utility $\frac{1}{T}\sum_{t=0}^{T-1}\sum_{k\in\calK}E^t_k$.
The only quadratic terms in the MIP involve the product of a binary and a (non-negative) real-valued variable, which can be linearized in standard fashion.\footnote{The ``big-$M$'' trick for linearizing the product of a binary vairable $I$ and a (non-negative) real-valued variable $x$: replace the product $x\cdot I$ with a real variable $y$, and add constraints $y \leq I M$, $y\geq 0$, $y\leq x$ and $y\geq x - (1-I)M$.}

\section{Empirical Evaluation}
\label{sec:experiments}

We present some simple proof-of-concept experiments to demonstrate the value that can be generated by prompting policies. We generate random problem instances, using relatively small numbers of users and providers to illustrate key points, and compare the user welfare (or utility) generated over time of specific policies. While large-scale experiments and empirically determined models of behavior and incentives will be needed for practical deployment, these results suggest that prompting has an important role to play in improving RS ecosystem health. 

\begin{figure}
  \centering
  \includegraphics[width=\columnwidth, trim={3.4cm 1.2cm 3.5cm, 3.1cm},clip  ]{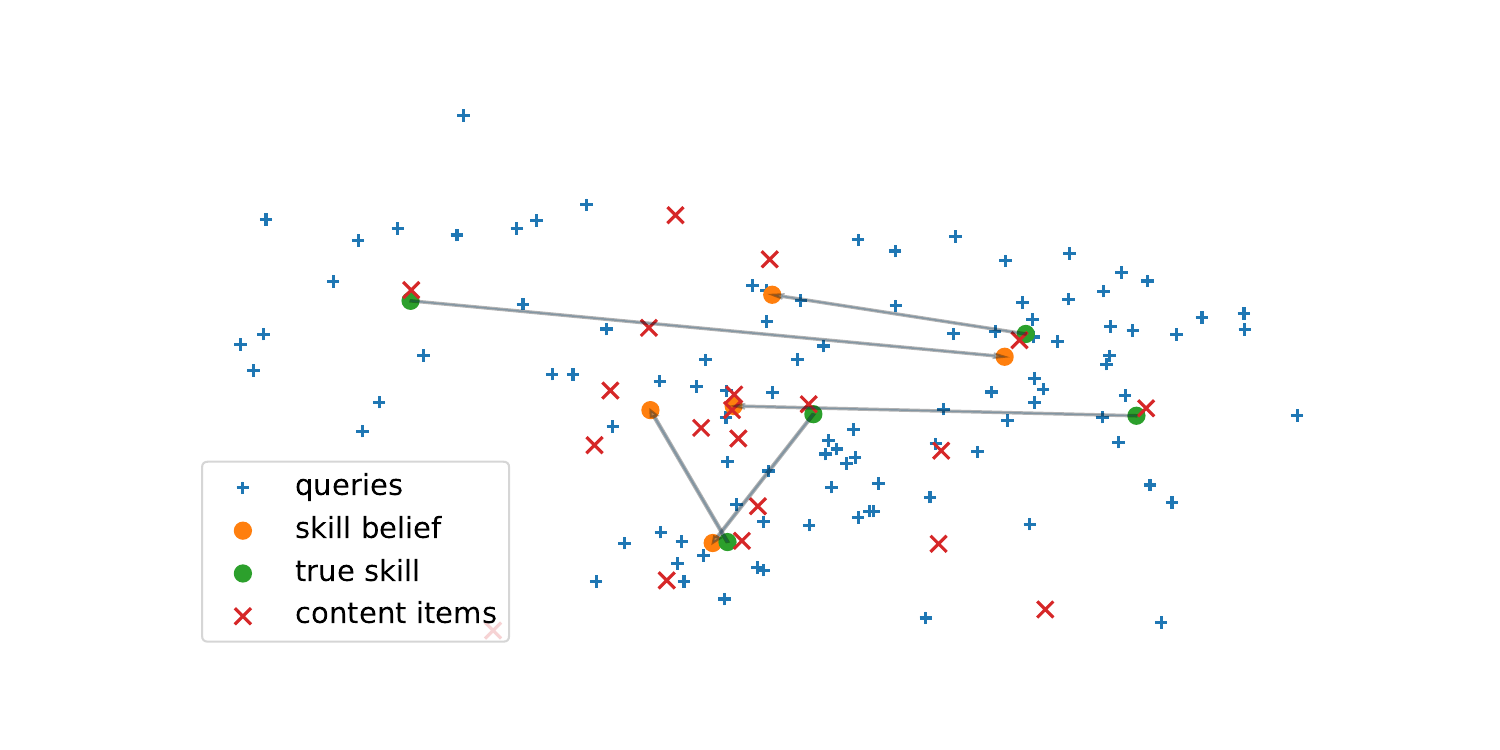}
  \caption{2-D problem instance illustration. Red crosses are content points (locations). Blue markers are user/query locations. Orange and green circles represent content providers: Green reflects a provider's true skills, while the connected orange circles capture's that providers skill beliefs.}
  \label{fig:illustrate}
\end{figure}

\vspace*{2mm}
\noindent
\textbf{Problem Generation}: 
Due to the infeasibility of testing our approach in a fully data-driven way, we generate synthetic (fully simulated) and semi-synthetic (preferences extracted from the MovieLens dataset~\cite{harper16:movielens}) scenarios. For each scenario, we generate multiple instances of problem size $(J, K, Q)$ with $J$ content points $\calJ$, $K$ providers $\calK$, and $Q$ users $\calQ$. Sizes vary across families of instance, but are kept small. 
In synthetic scenarios, instances are generated from a 
cluster model in 2-D space that emulates (implicit) communities of users with similar preferences. User content affinities decay linearly with Euclidean distance from a content point.
The semi-synthetic model adopts user and movie embeddings obtained by factorizing the MovieLens dataset. Movie embeddings are clustered with $k$-means to generate the possible content types. For each type $j$, the top $U_j$ users w.r.t.\ affinity are selected, where $U_j$ is determined based on the number of users with non-negative affinity to the cluster. (Affinities are inner products). Full data generation details 
are in Appendix~\ref{sec:app_exps}.  
See Fig.~\ref{fig:illustrate} for an illustration.

\vspace*{2mm}
\noindent
\textbf{Evaluation of the MIP}: We evaluate the MIP formulation on random small problem instances,
using various metrics, which we describe here:
\begin{itemize}
\item Let $E^T = \sum_{k\in\calK}E_k^T$ be the utility of the \emph{optimal prompting policy}, as determined by the MIP, at the final stage $T$, and $E = \frac{1}{T}\sum_{t=0}^{T-1}\sum_{k\in\calK}E_k^t$ be its \emph{time-averaged utility}---the latter is the MIP objective.
\item Let $\overline{E}^T$ and $\overline{E}$ denote the same quantities for the \emph{optimal policy that does not prompt providers}, but that does adapt its matching as providers update their locations (content).
\item Let $E_0$ be the utility of the \emph{optimal stationary policy} (no prompting, a fixed matching).
\item Define $P^T = (E^T - \overline{E}^T)/\overline{E}^T$ to be the \emph{final prompt gap}, i.e., the improvement obtained (at the final stage) by prompting, and $\widehat{P} = (E - \overline{E})/\overline{E}$ the time-averaged prompt gap, that is, the time-averaged improvement in utility obtained by prompting. 
\item Let $D = (\overline{E} - E_0)/E_0$ be the improvement obtained by the adaptive, non-prompting policy over the stationary policy.
\item Finally, let $$U_q = \frac{1}{T}\sum_{t=0}^{t-1}\sum_{k\in\calK}\pi^t_{q,k}\sum_{j\in\calJ}\Act^t_{k,j}s^\ast_{k,j}\sigma(q,j)$$ denote the {\em time-averaged user utility} of user $q$.
\end{itemize}

\begin{figure}[t]
    \centering
    \includegraphics[width=\linewidth, trim={0cm 0cm 0cm 1.15cm},clip]{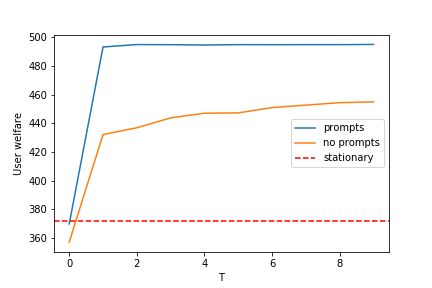}
    \caption{Total per-period utility $E^t$ (20 synthetic instances) for different policies ({\small $J = 20, K = 5, Q = 50, T = 10$}).}
    \label{fig:main_plot}
\end{figure}

\begin{figure}[t]
    \centering
    \includegraphics[width=\linewidth, trim={0cm 0cm 0cm 1cm},clip]{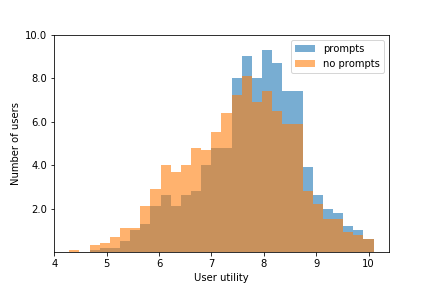}
    \caption{User utility $U_q$ counts (10 MovieLens instances) for different policies ({\small $J = 20, K = 10, Q = 100, T = 5$}).}
    \label{fig:main_histogram}
\end{figure}

We solve the MIP with Gurobi (the fastest commercial MIP solver).\footnote{\url{https://www.gurobi.com/}.} Fig.~\ref{fig:main_plot} plots per-period utility $E^t$ as a function of $t$ for each the described policies averaged over 20 random synthetic instances with $J=20$ content points, $K = 5$ providers, $Q = 50$ users over a time horizon of $T  = 10$, and Table~\ref{table:Q=50} contains key quantities for these instances in addition to a smaller set of instances with $J = 10$. The larger $(J = 20)$ instances exhibit larger final (9.3\%) and time-averaged (10.7\%) prompt gaps (roughly $2\times$) than the smaller $(J = 10)$ instances.

Table~\ref{table:movieworlds2} displays the same quantities averaged over 20 random MovieLens instances with $J\in\{10,20\}$, $K = 10$, $Q = 100$, and $T = 5$. Here, increasing the number of content locations seems to have a negligible impact on $P^T$ and $\hat{P}$, but the improvement over the stationary policy is significant (61.8\% for the $J = 20$ instances), showing the importance of dynamically matching and prompting based on the RS dynamics. The standard deviations of all quantities are fairly large, illustrating that for some instances the utility improvement due to prompting is not too large, but for other instances it is dramatic. Finally, as another means of visualization, Fig.~\ref{fig:main_histogram} displays the user utility histogram for the larger ($J=20$) MovieLens instances, showing a clear improvement due to prompting in both aggregate user utility and its distribution across users. The full set of experimental results is provided in Appendix~\ref{sec:app_exps}.

\begin{table}[t]
\small\addtolength{\tabcolsep}{-1pt}
\centering
\begin{tabular}{|c|c|c|c|c|c|c|}
\hline
$J$ &
  \begin{tabular}[c]{@{}l@{}}Avg $P^T$\end{tabular} &
  \begin{tabular}[c]{@{}l@{}}SD $P^T$ \end{tabular} &
  \begin{tabular}[c]{@{}l@{}}Avg $\widehat{P}$\end{tabular} &
  \begin{tabular}[c]{@{}l@{}}SD $\widehat{P}$\end{tabular} &
  \begin{tabular}[c]{@{}l@{}}Avg $D$\end{tabular} &
  SD $D$ \\ \hline\hline
10 &
  3.4\% &
  4.8\% &
  5.8\% &
  6.2\% &
  24.5\% &
  20.3\% \\ \hline
20 &
  9.3\% &
  7.2\% &
  10.7\% &
  7.5\% &
  18.2\% &
  9.0\% \\ \hline
\end{tabular}
\caption{Final gap $P^T$, time-avg.  gap $\widehat{P}$, stationary gap $D$ (20 synthetic instances), varying $J$ ({\small $K=5, Q=50, T = 10$}).}
\label{table:Q=50}
\end{table}

\begin{table}[t]
\small\addtolength{\tabcolsep}{-1pt}
\centering
\begin{tabular}{|c|c|c|c|c|c|c|}
\hline
$J$ &
  \begin{tabular}[c]{@{}l@{}}Avg $P^T$\end{tabular} &
  \begin{tabular}[c]{@{}l@{}}SD $P^T$ \end{tabular} &
  \begin{tabular}[c]{@{}l@{}}Avg $\widehat{P}$\end{tabular} &
  \begin{tabular}[c]{@{}l@{}}SD $\widehat{P}$\end{tabular} &
  \begin{tabular}[c]{@{}l@{}}Avg $D$\end{tabular} &
  SD $D$ \\ \hline\hline
10 & 3.3\% & 1.7\% & 3.3\%  & 1.7\% & 25.2\% & 14.2\% \\ \hline
20 & 3.7\% & 3.1\% & 3.5\% & 2.6\% & 61.8\% & 71.6\% \\ \hline
\end{tabular}
\caption{Final gap $P^T$, time-avg. gap $\widehat{P}$, stat.\ gap $D$ (10 MovieLens instances), varying $J$ ({\small $K\! =\! 10, Q\! =\! 100, T\! =\! 5$}).}
\label{table:movieworlds2}
\end{table}

We note that our MIP does not scale well beyond the instance sizes above. For example, on a MovieLens instance with $J = 50$, $K = 10$, $Q = 100$, $T = 5$, Gurobi was unable to find a feasible solution after 12 hours. Experimentally, the MIP scales poorly in the number of providers $K$ and the time horizon $T$. Improving our methods work at the scale of modern content ecosystems is a critical direction for future research. We discuss one such extension next.

\vspace*{2mm}
\noindent
\textbf{Column Generation}: In Appendix~\ref{sec:app_cg} we outline the initial ideas behind a \emph{column generation (CG)} approach that offers both a more general and more scalable solution to this complex \emph{multi-agent (i.e., multi-provider) planning problem}. The problem of ``moving'' providers through location space has strong analogies with multi-robot path planning~\citep{haghani_multirobot:arxiv21}, where the goal is determine the best deployment of a set of robots to specific tasks (e.g., in-warehouse order fulfillment), together with collision-free paths that allow them to complete their assigned tasks. While there are significant differences---providers are self-determining agents with their own goals and incentives---our CG formulation draws inspiration from the model of \citet{haghani_multirobot:arxiv21}.

\section{Concluding Remarks}
\label{sec:conclude}

We have developed a model of content provider incentives and behaviors that allows an RS to prompt providers to offer novel content items that help improve user social welfare and overall ecosystem health. Our prompting policies incentivize providers to ``explore'' w.r.t.\ their own skills and auidence beliefs, and \emph{de-risk} this exploration, nudging the RS to an equilibrium that improves user welfare and the utility of individual providers. Our prompting policies effectively break the fundamental information asymmetry that exists in many RS ecosystems. Apart from theoretical guarantees, our empirical results demonstrate that such prompting can significantly improve outcomes.

A number of important theoretical and practical extensions of this model are needed. Scalability is of critical importance---CG appears quite promising, but online adaptive policies using multi-agent RL should be investigated as well. Various extensions and generalizations of our model should prove valuable, including relaxing some of its more restrictive assumptions, such as:  the stationarity of user/query affinities and provider skills; the non-strategic decision making of providers (who best respond only myopically); and the extensive knowledge of the RS (e.g., of provider skills, user affinities). Of special interest is the case when the RS is uncertain of provider abilities (skill) and incentives (reward), which moves us into the realm of dynamic mechanism design \cite{bergemann:jel19}. A formal study of supply/demand information ``leakage'' embedded in prompts should generate useful insights as well. in Finally, the \emph{practical} prompting of providers requires translating the abstract notion of content points into \emph{actionable} prompts for providers (e.g., using generative modeling techniques to describe/suggest novel items), a topic of considerable importance.

\bibliography{long,standard,newrefs}

\appendix
\onecolumn
\section{Omitted Results and Proofs from Section~\ref{sec:oneprovider}}\
\label{sec:app_oneprovider}

\subsection{Prompting under Non-generalizing Belief Updates}\label{sec:app_oneprovider_nongeneralizing}

Given the fixed locations of all other providers, let $j^\ast\in\calJ$ be the optimal location for provider $k$, i.e., the point that maximizes long-term average reward.
Before addressing policy optimization in its full generality,
we first how the RS can move a (rationalizable) state $s^t = \langle L^t, \tA^t, \tS^t\rangle$ to $s^{t+1}=\langle L^{t+1}, \tA^{t+1}, \tS^{t+1}\rangle$, where for all providers $k' \neq k$, $\nu_{k'} = \texttt{None}$, and $k$'s location moves from $j$ to $j^*$
assuming specific audience and skill belief updates.

Let $\mu^t$ be a (non-myopically) stable matching w.r.t\ $s^t$,\footnote{Stability may be induced by earlier prompting (see next section), or by simple \emph{self-exploration} by providers (see discussion of exploration above).}
and let $M_{q, k}$ be the Bernoulli random variable denoting whether $q\in\calQ$ is matched to $k$, with $\Pr(M_{q, k}=1) = \mu^t(q,k)$.
We assume an utility function of the form $f(\sigma(q, j), s_{k, j}) = \sigma(q, j)\cdot s_{k, j}$.
Thus,
$k$'s audience and user utility under $\mu^t$ are the random variables (resp.):
$$A_{k,j} = \sum_{q\in\calQ}M_{q, k}\cdot\sigma(q, j);  \quad\quad
E_{k,j} = s_{k, j}\cdot A_{k,j} = \sum_{q\in\calQ}M_{q, k}\cdot \sigma(q, j)\cdot s_{k, j}.$$  

Provider $k$'s willingness to act based on an RS prompt $(j, C_{k,j})$ is governed by their trust $\lambda^t_k\in [0,1]$: $k$'s tentative audience belief is updated using incremental-trust belief update: $\widetilde{a}^{t+1}_{k,j} = (1 - \lambda^t_k)\cdot \widetilde{a}^{t + 1/2}_{k,j} + \lambda^t_k \cdot C_{k,j}$.
We define an \emph{accuracy-based} trust update as follows: 
$$\tau(\lambda_k, C_{k,j}, A_{k, j}) = \begin{cases} \max\{1,\lambda_k + \eta\cdot(1-|C_{k,j} - A_{k,j}|/|\calQ|)\} & \text{if } 1-|C_{k,j}-A_{k,j}|/|\calQ| > \lambda_k \\ \min\{0,\lambda_k - \eta\cdot(1-|C_{k,j} - A_{k,j}|/|\calQ|)\} & \text{if }1-|C_{k,j} - A_{k,j}|/|\calQ| < \lambda_k \\ \lambda_k & \text{if }1-|C_{k,j} - E_{k,j}|/|\calQ| = \lambda_k\end{cases}$$ where $\eta$ is a provider-specific learning rate. Informally, this trust update takes a step in the direction of the accuracy of the promised utility (with $1$ being perfect accuracy and $0$ the least accuracy).
Let
$$\phi(j; \lambda_k, \tS_k, \tA_k) = \max\left\{\frac{1}{\lambda_k}\left(\frac{\widetilde{s}_{k,\ell}\cdot\widetilde{a}_{k, \ell}}{\widetilde{s}_{k, j}} - (1-\lambda_k)\tA_{k,j}\right): \widetilde{s}_{k,\ell}\cdot\widetilde{a}_{k, \ell} > \widetilde{s}_{k,j}\cdot\widetilde{a}_{k,j}, \ell\in\calJ\right\}$$ 
denote the minimum level of utility the RS must guarantee to incentivize $k$ to move to $j$ given trust $\lambda_k$, skill beliefs $\widetilde{S}_k$, and audience beliefs $\tA_k$, given the form of incremental-trust tentative audience belief update.
Naturally, the minimum level of promised utility decreases as $\lambda_k$ increases:
\begin{lemma}
For any $j, \tS_k,\tA_k$, if $\lambda_k\ge\lambda_k'$ then $\phi(j; \lambda_k, \tS_k, \tA_k)\le\phi(j; \lambda_k', \tS_k, \tA_k)$.
\label{lemma:minincentive}
\end{lemma}

\begin{proof}
For a fixed $\ell$ such that $\widetilde{s}_{k,\ell}\cdot\widetilde{a}_{k,\ell} > \widetilde{s}_{k,j}\cdot\widetilde{a}_{k,j}$, we have $$\frac{\partial}{\partial\lambda}\left(\frac{1}{\lambda}\left(\frac{\widetilde{s}_{k,\ell}\cdot\widetilde{a}_{k,\ell}}{\widetilde{s}_{k,j}} - (1-\lambda)\widetilde{a}_{k,j}\right)\right) = -\frac{\widetilde{s}_{k,\ell}\cdot\widetilde{a}_{k,\ell} - \widetilde{s}_{k,j}\cdot\widetilde{a}_{k,j}}{\widetilde{s}_{k,j}\lambda^2} < 0,$$ so $\frac{1}{\lambda}(\frac{\widetilde{s}_{k,\ell}\cdot\widetilde{a}_{k,\ell}}{\widetilde{s}_{k,j}} - (1-\lambda)\widetilde{a}_{k,j})$ is decreasing in $\lambda$. Therefore the maximum over all $\ell$ such that $\widetilde{s}_{k,\ell}\widetilde{a}_{k,\ell} > \widetilde{s}_{k,j}\widetilde{a}_{k,j}$ is also decreasing in $\lambda\in[0,1]$.
\end{proof}

The main challenge the RS needs to address is the fact that a the provider might require a high degree of trust for a prompt to induce a location change. Consider first the case where $\mu^t$ is deterministic and $k$'s beliefs collapse immediately.
The following simple scheme allows the RS to induce a move from $j$ to $j^*$ (if at all possible) while keeping $\mu^t$ unchanged until the move occurs.

\begin{enumerate}
    \item \emph{Phase 1 (build trust):} While $\lambda_k < 1$, repeat the match-prompt pair $$(\mu,\nu) = \big(\mu^t, (j, \textstyle\sum_{q\in\calQ}\mu^t(q,k)\cdot\sigma(q,j))\big).$$
    \item \emph{Phase 2 (prompt to goal):} Issue match-prompt pair $$(\mu^t, (j^*, \phi(j^*; 1, \tS_k, \tA_k))).$$
\end{enumerate}

If $\phi(j^*; 1, \tS_k, \tA_k)\le\sum_{q\in\calQ}\mu^t(q,k)\cdot\sigma(q, j^*)$ then $\mu^t$ itself (stochastically) delivers the promised audience. If not, the provider trust might take a hit, but can be rebuilt later. We expand on this further consequently. 

\begin{thmdupl}[\ref{thm:reachstatedeterm}]
Let $s^t = \langle L^t, \tA^t, \tS^t\rangle$ be rationalizable and $\mu^t$ be a non-myopically stable matching w.r.t.\ $s^t$. After $\zeta=1/\eta$ iterations of the above policy, the resulting state $s^{t+\zeta}$ remains rationalizable and $\mu^t$ non-myopically stable w.r.t.\ $s^{t+\zeta}$. Furthermore, it maximizes the minimum trust, $\min_{t}\lambda^t$, among all policies that reach $s^{t+\zeta}$.
\end{thmdupl}

\begin{proof}
Trust is built up to $1$ in Phase 1 over the course of at most $1/\eta$ steps, since $\tau(\lambda_k^t, A_{k,j}, A_{k,j}) = \max\{1, \lambda_k^t + \eta\}$. In Phase 2 (a single step), $j^*$ is prompted, and accepted due to the promise of $\phi(j^*; 1, \tS_k, \tA_k)$, which collapses and induces equilibrium. No policy that induces this equilibrium can attain a higher value of $\min_t \lambda^t_k$. This is because the prompt in Phase 2 is the only step that can cause trust to decrease, and building up trust to $1$ in the previous phase ensures that the final trust value is as high as possible (by Lemma~\ref{lemma:minincentive}).
\end{proof}

We now relax the deterministic-matching and immmediate-belief-collapse assumptions. Instead, we assume a sample complexity measure $T(\varepsilon,\delta)\in\N$ that satisfies (letting $j\in\calJ$ and $\mu$ be s.t.\ $\mu^t = \mu$ for all $t$ where $\ell_k^t = j$):
$$|\{t : \ell_k^t = j\}|\ge T(\varepsilon,\delta)\implies \Pr\left(\forall t\ge\bar{t}\; |\widetilde{s}^t_{k,j}\cdot\widetilde{a}^t_{k,j} - s_{k,j}\cdot\E[A_{k,j}]|\le\varepsilon\right)\ge 1-\delta$$ where $\bar{t}$ is the latest stage where $\ell_k^{t} =j$. That is, if $k$ offers item $j$ 
enough,
its predicted utility (product of skill, audience beliefs) is nearly perfect (and remains so) with high probability. With only this assumption, the RS adopts the following policy to induce a move to $j^*$, parameterized by a trust threshold $\Lambda$ and an error threshold $\varepsilon$:

\begin{enumerate}
    \item \emph{Phase 1 (build trust):} While $\lambda_k < \Lambda$, repeat the match-prompt pair $$(\mu,\nu) = \left(\mu, \left(\ell_k, \textstyle\sum_{q\in\calQ}\mu(q,k)\cdot\sigma(q,\ell_k)\right)\right).$$ (Here $\ell_k$ is creator $k$'s location at the current time. We omit time in the superscripts for brevity.)
    \item \emph{Phase 2 (prompt to goal):} Issue match-prompt pair $$(\mu, (j^*,\phi(j^*; \Lambda, \tS_k, \tA_k))).$$
    \begin{enumerate}
        \item If Phase 2 has been entered $\ge T(\varepsilon, \delta/2)$ times, terminate.
        \item Else if $\lambda_k\ge \Lambda$ go to Phase 2.
        \item Else if $\lambda_k < \Lambda$ go to Phase 1.
    \end{enumerate}
\end{enumerate}

\begin{thmdupl}[\ref{thm:reachstategeneral}]
Assume $|\calQ|\ge 6\log 2/(1-\Lambda)^2$, let $s^t = \langle L^t, \tA^t, \tS^t\rangle$ be rationalizable, and $\mu^t$ be non-myopically stable w.r.t.\ $s^t$. If we run the above policy with $\varepsilon < \frac{1}{2}\min_{j, j'}|\E[E_{k,j}] -\E[E_{k,j'}]|$, the following hold with probability at least $1-\delta$: (1) the resulting state $s^{t+\zeta}$ has $\ell_k = j^*$, and remains rationalizable, while $\mu$ remains non-myopically stable w.r.t.\ $s^{t+\zeta}$; and (2) the policy terminates in $\zeta=O\big(\frac{T(\varepsilon,\delta)^3\Lambda^2}{\eta^2\delta^2}\big)$ rounds.
\end{thmdupl}

\begin{proof}
We show that with high probability, at most a quadratic number of time steps are spent in Phase 1 before the creator's trust exceeds the threshold $\Lambda$. Since the prompts in Phase 1 mirror the content item $\ell_k^t = BR(\tS^t_k, \tA^t_k)$ and since $\mu$ is non-myopically stable, the creator is guaranteed to accept each Phase 1 prompt. Furthermore, due to stability, a single content item $j$ achieves $BR(\tS^t_k, \tA^t_k)$ for all $t$ in the absence of prompts. For simplicity, we will assume that the creator's beliefs for item $j$ are already $\varepsilon$-accurate, that is, $|\tS_{k,j}\cdot\tA_{k,j} - S_{k,j}\cdot\E[A_{k,j}]|\le \varepsilon$. If this is not the case, we simply add a waiting phase at the very beginning of the policy that allows for this belief to reach the desired accuracy (which would only affect the constants in the convergence rate).

Let $\lambda_k^t < \Lambda$ be the creator's trust at the beginning of some execution of Phase 1 (if trust is larger than $\Lambda$, we would not be in Phase 1). Observe that \begin{align*}\frac{|A_{k,j} - \E[A_{k,j}]|}{|\calQ|} \le \frac{1-\lambda_k^t}{2} &\implies \frac{|A_{k,j} - \E[A_{k,j}]|}{|\calQ|} \le \min\left\{1-\lambda_k^t, \frac{1}{2}\right\} \\ &\implies 1 - \frac{|A_{k,j} - \E[A_{k,j}]|}{|\calQ|} \ge \max\left\{\lambda_k^t, \frac{1}{2}\right\} \\ &\implies \tau(\lambda_k, \E[A_{k,j}], A_{k,j}) \ge \lambda_k^t + \frac{\eta}{2}.\end{align*} Therefore $$\Pr\left(\tau(\lambda_k^t, \E[A_{k,j}], A_{k,j}) \ge \lambda_k^t + \frac{\eta}{2}\right) \ge \Pr\left(|A_{k,j} - \E[A_{k,j}]| \le \frac{(1-\lambda_k^t)|\calQ|}{2}\right)\ge 1- 2\exp\left(-\frac{(1-\lambda_k^t)^2|\calQ|}{2}\right)$$ by a Chernoff-Hoeffding bound. For $|\calQ|\ge\frac{6\log 2}{(1-\Lambda)^2}$ this probability is $\ge 3/4$. Therefore, every time $(j, \E[A_{k,j}])$ is given as a prompt to the provider in Phase 1, trust increases by at least $\eta/2$ with probability $\ge 3/4$, and trust decreases by at most $\eta$ with probability $\le 1/4$ (if trust is already at $0$, it increases by at least $\eta/2$ with probability $\ge 3/4$, and stays at $0$ with probability $\le 1/4$). We show that with high probability, trust $\lambda_k^t$ exceeds $\Lambda$ in quadratically many time steps. To prove this, we bound the hitting time of the trust process $(\lambda_k^t)$ by the hitting time of a simple random walk on $\Z$ via a standard coupling argument. For $s\in\N$ let $X_s$ be $+1$ with probability $1/2$ and $-1$ with probability $1/2$. Consider the process $(Y_t)$ defined as follows: every time $(X_s)$ walks one unit to the right, so does $(Y_t)$. If $(X_s)$ walks one unit to the left, we have two cases: (1) If $(X_s)$ walks one unit to the right in the subsequent time step, then $(Y_t)$ walks one unit to the right. (2) If $(X_s)$ walks one unit to the left in the subsequent time step, then $(Y_t)$ walks two units to the left. Then, $\Pr(Y_t = +1) = 3/4$ and $\Pr(Y_t = -2) = 1/4$, and by how we defined the coupling, $$\inf\left\{ t : \textstyle\sum_{t'\le t} Y_{t'}\ge k\right\}\le\inf\left\{s : \textstyle\sum_{s'\le s}X_{s'}\ge k\right\}$$ for any $k$. So, coupling the trust process $(\lambda_k^t)$ to $(Y_s)$ in the obvious way, we can bound the the hitting time $\tau_{\Lambda}$ of the trust process: $$\tau_{\Lambda} := \inf\{ t : \lambda_k^t \ge \Lambda\} = \inf\left\{t : \textstyle\sum_{t'\le t}Y_{t'}\ge \frac{2\Lambda}{\eta}\right\}\le\inf\left\{s: \textstyle\sum_{s'\le s} X_{s'}\ge \frac{2\Lambda}{\eta}\right\}.$$ Classical results~\citep{levin2017markov} yield $$\Pr(\tau_{\Lambda} \ge\kappa) \le \Pr\left(\inf\left\{t: \textstyle\sum_{t'\le t} X_{t'}\ge \frac{2\Lambda}{\eta}\right\}\ge \kappa\right)\le \frac{6\cdot(2\Lambda/\eta)}{\sqrt{\kappa}},$$ and so $$\Pr(\tau_{\Lambda}\le\kappa\text{ for every execution of Phase 1})\ge 1 - \frac{12\Lambda T(\varepsilon,\delta/2)}{\eta\sqrt{\kappa}}$$ by a union bound. Taking $\kappa\ge 576\Lambda^2T(\varepsilon,\delta/2)^2/(\eta^2\delta^2)$ yields $\Pr(\tau_{\Lambda}\le\kappa\text{ for every execution of Phase 1})\ge 1-\delta/2$. We enter Phase 1 at most $T(\varepsilon,\delta/2)$ times, so with probability $\ge 1-\delta/2$ the total number of time steps spent in Phase 1 is $O(T(\varepsilon,\delta/2)^3(\Lambda/\eta\delta)^2)$.

Finally, the prompting policy terminates once $j^*$ has been visited at least $T(\varepsilon,\delta/2)$ times, which means the creator's believed engagement for $j^*$ is $\varepsilon$-accurate throughout the execution of the policy with probability at least $1-\delta/2$. Therefore, since $\varepsilon < \frac{1}{2}(\E[E_{k, j^*}] - \E[E_{k, j}])$, $j^* = BR(\tS_k, \tA_k)$ at the end of the policy with probability at least $1-\delta/2$. This, combined with the above high-probability guarantee on the number of iterations spent in Phase 1, completes the proof. 
\end{proof}

By varying $\Lambda$, our policy interpolates between (1) reaching the optimal equilibrium as quickly as possible and (2) maintaining the greatest trust at each step. Indeed, $\Lambda_{*} = \min\{\lambda : \phi(j^*; \lambda, \tS^t_k, \tA^t_k)\le Q\}$ is the least trust threshold s.t.\ the RS successfully prompts $k$ to offer $j^*$ in Phase 2. The hitting time of $\Lambda$ in Phase 1 is increasing in $\Lambda$, so $\Lambda_{*}$ minimizes (in expectation) the number of steps needed to reach equilibrium. As $\Lambda\ge \Lambda_{*}$ increases, the post-prompt drop in provider trust in Phase 2 decreases, but the number of steps needed to reach equilibrium increases. 


\subsection{Prompting when Beliefs Generalize across Content}\label{sec:app_oneprovider_generalizing}

To abstract away the details of provider dynamics, we assume deterministic matchings, and that $k$'s skill and audience beliefs for any point $j$ collapse to their true values the first time $j$ is produced (this can be relaxed with only minor modifications to our argument). We now allow observations at $j$ to affect beliefs at other points as well. We assume only the following: 

\begin{enumerate}
    \item Beliefs at unvisited points never become less accurate, i.e., if $\widetilde{s}^t_{k,j} \ge s_{k, j}$ (resp. $\widetilde{s}^t_{k,j}\le s_{k,j}$) and some $j'\neq j$ is produced at time $t$, then $\widetilde{s}^{t}_{k,j} \ge \widetilde{s}^{t+1}_{k, j}\ge s_{k,j}$ (resp. $\widetilde{s}^t_{k,j}\le\widetilde{s}^{t+1}_{k,j}\le s_{k,j}$), and similarly for audience beliefs.
    \item Belief updates for unvisited points are independent of the order in which other points are visited; specifically for any $j, j'$, the accuracy increase in $k$'s predicted reward at $j'$ if $j$ is visited is independent of other visited points.
\end{enumerate}

For $\calJ'\subseteq\calJ$, let $$I(\calJ') = \{j\in\calJ : \phi(j; 1, \tS_k, \tA_k)\le Q\} \supset \calJ'$$ be the set of all incentivizable points given that $k$ has visited all and only points in $\calJ'$ ($I(\calJ')$ is well-defined and unique given our second assumption). 
We make the following additional assumption: for any unvisited $j, j'\in\calJ\setminus\calJ'$, visiting $j$ does not make $j'$  harder to incentivize, i.e., $$I(\calJ_1\cup\calJ_2)\supseteq I(\calJ_1)\cup I(\calJ_2)$$ for any $\calJ_1, \calJ_2\subseteq\calJ$. We can readily determine if a point $j^\ast$ is (or can be made to be) incentivizable. Let $I^m(\calJ') = I(I(\cdots I(\calJ')\cdots ))$ (closure under $m$ iterations) for any $m\leq J$. All points in this set are incentivizable with no more than $|I^m(\calJ')|$ prompts:

\begin{lemma}\label{lemma:incentivizable}
Given initial beliefs $\tS^0_k, \tA^0_k$ and a target content point $j^\ast$, determining if  $j^\ast \in I^J(\calJ')$, i.e., whether $j^\ast$ can be incentivized by prompting, can be done in polynomial time.
\end{lemma}

\begin{proof}
Let $\calJ_0$ be the set of undominated overestimates that will be produced by the provider without any prompting by the RS. For each $t\ge 1$, compute $\calJ_{t} = I(\calJ_{t-1})$, and stop when $\calJ_{t} = \calJ_{t-1}$. If $j^*\in \calJ_t$, then there is a sequence of content items that leads to $j^*$ being incentivizable, otherwise no such path exists. 
\end{proof}

If $j^\ast$ is incentivizable, there must be a prompting policy that induces $k$ to visit $j^\ast$. In one naive such policy, the RS prompts $k$ to visit all points in $I(\calJ')$ in some order, i.e., along a specific \emph{content path}, then prompts visits to all points in $I(I(\calJ'))$, and so on, until $j^*$ belongs to the incentivizable set (at which point it is prompted). Of course, this may induce visits to unneeded points. The shortest promptable content path can be formulated as the following mixed integer program (MIP).

Let $m(j)$ be the minimal $m$ such that $j\in I^{m}(\calJ_0)$, where $\calJ_0$ is the initial set of content items (undominated overestimates) that $k$ visits without RS prompts. Consider the following weighted directed graph with vertex set $I(\calJ_0)\cup I^2(\calJ_0)\cup\cdots\subseteq\calK$: the edges in this graph are of the form $(j, j')$ for every $m(j)\le m(j')$, and the weight of edge $(j, j')$, denoted $w_{j, j'}$, is the increase in accuracy in the provider $k$'s believed reward for $j'$ when $j$ is visited. Finally, let $r_{j'}$ denote the minimum believed reward ($\widetilde{s}_{k, j'}\cdot\widetilde{a}_{k, j'}$) needed to incentivize $j'$ given that $\calJ_0$ is the only set of content that has been produced (this is the ``lowest bar" for incentivizability, and the assumption that $I(\calJ_1\cup \calJ_2)\supset I(\calJ_1)\cup I(\calJ_2)$ ensures that raising predicted belief above this threshold suffices to render $j'$ incentivizable). Given this, we can formulate the shortest feasible content path problem as a binary integer program. The variables of the IP are of the form $x_{j, j'}$ for every content points $j, j'\in\calK$ with either $m(j) = m(j')$ or $m(j)+1 = m(j')$ (any path that moves from a given level to a lower level can be replaced by a feasible path of equal length where the path traverses levels in nondecreasing order, due to the assumption that information gain on predicted rewards is order-independent). Let $j^*$ denote the target content point. The integer program is:

\[\arraycolsep=1.4pt\def\arraystretch{1.5}
\begin{array}{ll}\text{minimize }&  \sum_{j_1, j_2\in\calK} x_{j_1, j_2} \\
\,\text{subject to } 
&\sum_{j\in\calK} x_{j, j^*} = 1 \\
& \sum_{j_2\in N(j_1)}x_{j_1, j_2} + x_{j_2, j_1} = 2\hfill\forall j_1\in\calK\setminus\{j^\ast\} \\
&\sum_{m(j_1')\le m(j_2)} w_{j_1', j_2}\Big(\sum_{j_2'\in\calK}x_{j_1', j_2'}\Big) \ge r_{j_2}\cdot x_{j_1, j_2}\qquad \forall j_1, j_2\in\calK \\
& x_{j_1, j_2} \in \{0, 1\} \hfill \forall j_1, j_2\in\calK
\end{array}
\]

The first two constraints ensure that the solution defines a valid path through the content graph that ends at $j^*$. The third
third constraint ensures feasibility of the path with respect to incentivizability (if $j_2$ is visited after $j_1$, the sum of the accuracy gains on predicted belief at $j_2$ thus far must exceed the incentivizability threshold $j_2$).

While this MIP finds the shortest promptable path, the problem is NP-hard (by reduction from weighted-constrained shortest path). Fortunately, there is a simple, efficient greedy algorithm: at each step, if $j^\ast$ is incentivizable, prompt with $j^\ast$ and terminate; otherwise prompt with the (currently) incentivizable point $j$ that maximizes $w_{j,j^*}$, the increase in accuracy in $k$'s predicted reward at $j^\ast$ if $j$ is visited. Let $\textsc{Greedy}$ denote the length of the greedy path, and let $\OPT$ denote the length of the shortest path.

\begin{thmdupl}[{\ref{thm:greedy}}]
Let $r$ denote the minimum predicted reward threshold such that if $\widetilde{s}_{k, j^*}^0\cdot\widetilde{a}_{k, j^*}^0 \ge r$, then $j^*\in I(\emptyset)$ (that is, $j^*$ would be immediately incentivizable). Let $w_{min}$ be the minimum $w_{j,j^\ast}$ for any point $j$ in the shortest path. Then $\textsc{Greedy}\le \frac{r}{w_{min}}\cdot\OPT$.
\end{thmdupl}

\begin{proof}
Consider the first point the greedy algorithm deviates from the shortest content path. Let $j$ be the content item chosen by the shortest path. The length of the greedy content path until it either reaches $j^*$ or rejoins the shortest content path at $j$ is at most $\frac{r}{w_{j, j^*}}$, since every content item $j'$ on the greedy content path satisfies $w_{j', j^*} > w_{j, j^*}$. The greedy content path can deviate from and rejoin the shortest content path at most $\OPT$ times, so its length is at most $$\sum_{j\in\widetilde{\calJ}}\frac{r}{w_{j, j^*}}\le \frac{r}{w_{min}}\cdot\OPT,$$ where $\widetilde{\calJ}$ denotes the set of content items on the shortest content path where the greedy content path deviates. 
\end{proof}

Lemma~\ref{lemma:incentivizable} and Thm.~\ref{thm:greedy} together provide a polynomial-time procedure for computing a prompting policy for $k$ that moves it to the optimal incentivizable content point along a content path not much longer than the shortest such path.

\section{Omitted Details and Mixed Integer Programming Formulations in Section~\ref{sec:joint}}
\label{sec:app_mips}

In this appendix section, we provide full details of the joint provider prompting MIP formulation.

\subsection{MIP Formulation}

We provide the full MIP formulation that computes the optimal joint-provider prompting policy below:

\[\arraycolsep=1.4pt\def\arraystretch{2}
\begin{array}{ll}\text{maximize }&  \frac{1}{T}\sum_{t=0}^{T-1}\sum_{k\in\calK} E_k^t \\
\,\text{subject to } 

&(1a)\; C_{k,j}^t\le \nu_{k,j}^tM\hfill \forall k\in\calK, j\in\calJ, t< T\\
&(1b)\; \Aud_{k}^t\ge C_{k,j}^{t-1}\Act_{k,j}^t\hfill\forall k\in\calK, j\in\calJ, t< T \\

&(2)\; \widetilde{s}^t_{k,j} \widetilde{a}^t_{k,j}\ge \widetilde{s}^t_{k,j'} \widetilde{a}^t_{k,j'} - (1 - \Act_{k,j}^t)M \hfill \forall k\in\calK, j\in\calJ, j'\in\calJ, t< T\\

&(3)\; \widetilde{a}^t_{k,j}\lesseqgtr \Aud_k^{t-1} \pm (1-\Act_{k,j}^{t-1})M \pm \nu_{k,j}^{t-1}M \hfill \forall k\in\calK, j\in\calJ, t< T \\

&(4)\; \widetilde{a}^t_{k,j}\lesseqgtr \widetilde{a}^{t-1}_{k,j} \pm \Act_{k,j}^{t-1}M \pm \nu_{k,j}^{t-1}M \hfill \forall k\in\calK, j\in\calJ, t< T \\

&(5)\; \widetilde{a}^t_{k,j}\lesseqgtr C_{k,j}^{t-1} \pm (1-\nu_{k,j}^{t-1})M \hfill \forall k\in\calK, j\in\calJ, t < T \\

&(6)\; E_k^t = \sum_{q\in\calQ}\pi_{q,k}^t\sum_{j\in\calJ}\Act^t_{k,j}s^\ast_{k,j}\sigma(q,j) \hfill \forall k\in\calK, t< T \\

&(7)\; \Aud_{k}^t = \sum_{q\in\calQ}\pi_{q,k}^t\sum_{j\in\calJ}\Act_{k,j}^t\sigma(q,j)\hfill\forall k\in\calK, t< T \\ 

&(8)\; \widetilde{s}^t_{k,j} = \widetilde{s}^0_{k,j}(1-\Vis_{k,j}^{t-1}) + s^\ast_{k,j}\Vis_{k,j}^{t-1}\hfill\forall k\in\calK, j\in\calJ, 1\le t < T \\ 

&(9a)\; \Vis_{k,j}^t\ge\Act_{k,j}^{t'}\hfill\forall k\in\calK, j\in\calJ, 1\le t < T, t' \le t \\

&(9b)\; \Vis_{k,j}^t\le \sum_{t' = 0}^{t}\Act_{k,j}^{t'} \hfill \forall k\in\calK, j\in\calJ, 1\le t < T \\

&(10)\; \sum_{j\in\calJ}\Act_{k,j}^t = 1\hfill\forall k\in\calK, t < T \\

&(11)\; \sum_{k\in\calK} \pi_{q, k}^t = 1\hfill\forall q\in\calQ, t < T \\

&(12)\; \sum_{j\in\calJ} \nu_{k,j}^t\le 1\hfill\forall k\in\calK, t < T \\

&(13)\; \pi_{q,k}^t \in [0, 1], \nu_{k, j}^t\in\{0,1\}, C_{k,j}^t\ge 0, \Act_{k,j}^t\in\{0,1\}, \Vis_{k,j}^t\in\{0,1\} \qquad \forall k\in\calK, j\in\calJ, q\in\calQ, t < T

\end{array}
\]

Constraints 1--5 are discussed in Section~\ref{sec:joint}. Constraint 6 is the definition of welfare, Constraint 7 is the definition of audience, Constraint 8 enforces collapsing skill belief updates, Constraints 9a and 9b enforce that $\Vis_{k,j}^t = \max\{\Act_{k,j}^{t'}: t'\le t\}$, Constraint 10 ensures that a provider is at exactly one location at any given time, Constraint 11 ensures that the matching determines a probability distribution over providers, and Constraint 12 ensures that the RS provides at most one nudge to any given provider at any given time.

The constants (generated as described in the next section and given as input to the MIP) in the above formulation are (1) affinities $\sigma(q, j)$ for all $q\in\calQ, j\in\calJ$, (2) true skills $s^\ast_{k,j}$ for all $k\in\calK, j\in\calJ$, (3) initial skill beliefs $\widetilde{s}_{k,j}^0$ for all $k\in\calK, j\in\calJ$, and (4) initial audience beliefs $\widetilde{a}_{k,j}^0$ for all $k\in\calK, j\in\calJ$.

\section{Omitted Details and Results in Section~\ref{sec:experiments}}
\label{sec:app_exps}

Here we describe details pertaining to problem instance generation used in our experiments, present and describe the full set of experiments we ran, and include some initial thoughts on a column generation (CG) approach to the multi-provider planning problem.

\subsection{Problem Generation for Synthetic Data}
 We first describe the process used to generate random synthetic problem instances. Content and user embeddings lie in $\R^d_+$ (we use $d=2$). Our $J$ content points are sampled from a $d$-dimensional Gaussian $P_\calJ = \calN(\mathbf{0}, a\mathbf{I})$, where $a$ is a scalar.

Provider $k$'s skills are stationary and are represented by points $k\in\R^d_+$ in the same space (we equate provider $k$ with its ``skill point''). A provider is sampled from a Gaussian mixture with one component per content point $j\leq J$, where the mixture coefficient for component $i$ proportional to $p_C(c_i)$. 
For any content point $j$, $k$'s true skill is $-\frac{d(k, j) - \max_{j'} -d(k, j')}{\min_{j'} -d(k, j')}$. 
To generate initial skill beliefs $\widetilde{S}^0_k$ for provider $k$, a tentative``confused skill point'' is sampled from $\calJ$ with probability proportional to its (inverse) distance to their true skill $k$. Gaussian noise is added to determine their initial ``confused skill point'' $k_c$. Provider $k$'s skill belief $\widetilde{s}^0_{k,j}$ is determined analogously to true skill, but w.r.t.\ its confused skill $k_c$. 

Users $q\in\R^d_+$ are sampled from a Gaussian mixture whose components are centered on providers in a similar fashion. User affinities for content points are inversely proportional to their distance from that point: 
$\sigma(q,j) = G - d(q, j)$, where $G$ is a global constant ensuring affinities are positive. We keep $Q$
small, so interpret each $q\in\calQ$ as a user type/cluster.
\subsection{Problem Generation for Semi-Synthetic/MovieLens Data}
Problems in this scenario were generated by drawing from a pool of embeddings generated by factorizing the MovieLens dataset. First, a subset of approximately 100000 users and 1000 movies is used to generate a binary affinity matrix by conditioning on whether the users' movie rating is more than or equal to 4. The resulting sparse matrix is factorized using weighted alternating least squares (WALS). The movies embeddings are clustered with k-means into 200 clusters. Each cluster's center becomes a possible item.

For each individual scenario, a subset of $I$ items is sampled from the cluster centers without replacement, with probabilities proportional to the size (number of movies) of each cluster. For each sampled item, the set of users with affinity greater than $0.005$ is counted to estimate the popularity of the item. To generate $U$ users, we randomly choose to attach each one of them to one of the items with probability proportional to the item popularity as computed above. We then end up with a partition $U$ into $U_I$ classes. For each item $i$, we extract the top $U_i$ (in terms of affinity) users from the dataset and their corresponding embeddings. Affinities between user-item pairs are computed based on inner product and can be between 0 and 1 due to the original binarization of the matrix. Providers are generated based on sub- or super- sampling a set of items and adding uniform Gaussian noise. Since we have no data-driven way to generate confused beliefs, we then apply the confusion process of the synthetic scenario.

\subsection{Additional Experiments}

We include the full set of experiments we ran. These include 
\begin{itemize}
    \item Synthetic instances of sizes
    \begin{itemize}
        \item $(J\in\{10, 15, 20, 25\}, K = 5, Q = 30, T = 10)$ (Table~\ref{table:Q=30}, Figure~\ref{fig:Q=30})
        \item $(J\in\{10, 20\}, K = 5, Q = 50, T = 10)$ (Table~\ref{table:Q=50}, Figure~\ref{fig:Q=50})
        \item $(J\in\{10, 20, 30, 40\}, K = 10, Q = 50, T = 5)$ (Table~\ref{table:bigworlds}, Figure~\ref{fig:bigworlds})
    \end{itemize}
    \item Semi-synthetic/MovieLens instances of sizes
    \begin{itemize}
        \item $(J\in\{10, 20, 30, 40\}, K = 10, Q = 50, T = 5)$ (Table~\ref{table:movieworlds1}, Figure~\ref{fig:movieworlds}, Figure~\ref{fig:movieworlds_utils})
        \item $(J\in\{10, 20\}, K = 10, Q = 100, T = 5)$ (Table~\ref{table:movieworlds2}, Figure~\ref{fig:movieworlds2}, Figure~\ref{fig:movieworlds_utils2}).
    \end{itemize}
\end{itemize}

We ran Gurobi with a relative MIP gap of zero (guaranteeing the optimal solution) for only the smallest set of instances with $Q = 30$. For all others, we set the relative MIP gap to 0.05\%, meaning the solution returned is only guaranteed to be within $0.05\%$ of the truly optimal solution's objective value. We did not explicitly linearize our formulation as Gurobi itself recognizes and performs such linearizations in an optimized fashion.

Fig.~\ref{fig:Q=30} plots per-period utility $E^t = \sum_{k\in\calK}E_k^t$ as a function of $t$ for each the described policies averaged over 20 random instances (each) with $K = 5$ content providers, $Q = 30$ user queries, $T = 10$ time steps, and $J\in\{10, 15, 20, 25\}$ content points. We see that for the smaller instances, namely $J = 10$ and $J = 15$, the utility of the non-prompting policy is able to eventually approach the utility of the prompting policy. But for the larger instances, namely $J = 20$ and $J = 25$, the non-prompting policy has a tougher time increasing its per-period utility over the $T$ time steps. This is further illustrated in Table~\ref{table:Q=30}, which records the quantities (averaged over the $20$ instances) $P^T$, $\widehat{P}$, and $D$. The final prompt gap $P^T$ distinctly increases with the size of the instances, and the time-averaged prompt gap $\widehat{P}$ of the larger instances ($J\in\{15, 20, 25\}$) is twice that of the smallest instance ($J = 10$). The improvement over the stationary policy is significant; ranging from $40\% - 60\%$, showing the importance of dynamically matching and prompting based on the provider dynamics. Finally, the standard deviations of all quantities are fairly large, illustrating that for some instances the utility improvement due to prompting is not too large, but for other instances it is dramatic. We also solved the compact MIP for a set of slightly larger instances with $Q = 50$ user queries. For these instances, we terminated Gurobi once a relative MIP gap of $0.05\%$ was obtained. This means that the plotted utilities are within $0.05\%$ of the truly optimal utilities (though in most cases they are in fact optimal; for only a few of the instances Gurobi was unable to close the MIP gap completely even after 8 hours). Fig.~\ref{fig:Q=50} plots per-period utility averaged over 20 random instances (each) with $K = 5$, $Q = 50$, $T = 10$, and $J\in\{10, 20\}$, and Table~\ref{table:Q=50} records $P^T$, $\widehat{P}$, and $D$. The policies displayed similar trends on these larger instances: the improvements due to prompting are more significant for the instances with $J = 20$, and for these instances the per-period utility of the non-prompting policy has a tougher time approaching that of the prompting policy over the $T$ periods. Similar trends arise when the number of providers and items are bumped up to $K = 10$ and $J\in\{10, 20, 30, 40\}$ (Fig.~\ref{fig:bigworlds} and Table~\ref{table:bigworlds}).

On the MovieLens instances, larger numbers of content items has less of an effect on the (final and time-averaged prompt gap), and the gaps are generally smaller than those of the synthetic instances. But, increasing the number of content items has a marked impact on $D$, the improvement over the stationary policy (Tables~\ref{table:movieworlds1} and~\ref{table:movieworlds2}). Furthermore, unlike the synthetic instances where the non-prompting policy appears to improve over time, the non-prompting policy doesn't exhibit any discernible improvement over the 5 time steps, accruing significant user regret. Finally, for the MovieLens instances, we generated histograms of user utility counts to compare time-averaged utilities $U_q$ of individual users between the prompting and non-prompting policy (Figures~\ref{fig:movieworlds_utils} and~\ref{fig:movieworlds_utils2}). These reflect a clear improvement due to prompting---reducing the number of low-utility users and increasing the number of high-utility users. 

\textbf{Scalability}: We generally observed that the MIP suffered from scalability issues. One explanation is the large number of ``big-M'' constraints---Constraints~\ref{eq:con1a},~\ref{eq:con2},~\ref{eq:con3},~\ref{eq:con4},~\ref{eq:con5} all contain a term $M$ that increases with $Q$ (we define it to be $Q$ times the maximum query-content affinity). This inherently makes the dual bounds in branch-and-bound significantly weaker, slowing the search dramatically. One avenue towards a scalable solution is the column generation formulation we discuss next.

\begin{table}[t]
\centering
\begin{tabular}{|c|c|c|c|c|c|c|}
\hline
$J$ &
  \begin{tabular}[c]{@{}l@{}}Avg $P^T$\end{tabular} &
  \begin{tabular}[c]{@{}l@{}}SD $P^T$ \end{tabular} &
  \begin{tabular}[c]{@{}l@{}}Avg $\widehat{P}$\end{tabular} &
  \begin{tabular}[c]{@{}l@{}}SD $\widehat{P}$\end{tabular} &
  \begin{tabular}[c]{@{}l@{}}Avg $D$\end{tabular} &
  SD $D$ \\ \hline\hline
10 & 1.7\% & 2.5\% & 5.2\%  & 4.3\% & 39.8\% & 24.9\% \\ \hline
15 & 6.6\% & 5.8\% & 10.5\% & 5.7\% & 54.1\% & 26.5\% \\ \hline
20 & 8.2\% & 5.1\% & 9.7\%  & 5.6\% & 58.7\% & 41.2\% \\ \hline
25 & 9.2\% & 5.3\% & 10.1\% & 5.7\% & 40.6\% & 25.1\% \\ \hline
\end{tabular}
\caption{(Synthetic) Final gap $P^T$, time-avg. gap $\widehat{P}$,  stationary gap $D$ (avg. over 20 instances), varying $J$ ($K=5, Q=30, T = 10)$.}
\label{table:Q=30}
\end{table}

\begin{table}[t]
\centering
\begin{tabular}{|c|c|c|c|c|c|c|}
\hline
$J$ &
  \begin{tabular}[c]{@{}l@{}}Avg $P^T$\end{tabular} &
  \begin{tabular}[c]{@{}l@{}}SD $P^T$ \end{tabular} &
  \begin{tabular}[c]{@{}l@{}}Avg $\widehat{P}$\end{tabular} &
  \begin{tabular}[c]{@{}l@{}}SD $\widehat{P}$\end{tabular} &
  \begin{tabular}[c]{@{}l@{}}Avg $D$\end{tabular} &
  SD $D$ \\ \hline\hline
10 & 5.6\% & 4.5\% & 6.2\%  & 4.0\% & 25.9\% & 13.4\% \\ \hline
20 & 7.1\% & 3.1\% & 6.9\% & 2.7\% & 20.8\% & 12.3\% \\ \hline
30 & 8.8\% & 3.6\% & 8.2\%  & 3.4\% & 27.0\% & 13.8\% \\ \hline
40 & 8.9\% & 3.7\% & 8.4\% & 3.0\% & 27.7\% & 12.6\% \\ \hline
\end{tabular}
\caption{(Synthetic) Final gap $P^T$, time-avg. gap $\widehat{P}$, stationary gap $D$ (avg. over 20 instances), varying $J$ ($K=10, Q=50, T=5)$.}
\label{table:bigworlds}
\end{table}

\begin{table}[t]
\centering
\begin{tabular}{|c|c|c|c|c|c|c|}
\hline
$J$ &
  \begin{tabular}[c]{@{}l@{}}Avg $P^T$\end{tabular} &
  \begin{tabular}[c]{@{}l@{}}SD $P^T$ \end{tabular} &
  \begin{tabular}[c]{@{}l@{}}Avg $\widehat{P}$\end{tabular} &
  \begin{tabular}[c]{@{}l@{}}SD $\widehat{P}$\end{tabular} &
  \begin{tabular}[c]{@{}l@{}}Avg $D$\end{tabular} &
  SD $D$ \\ \hline\hline
10 & 3.3\% & 2.8\% & 3.1\%  & 2.4\% & 29.9\% & 14.1\% \\ \hline
20 & 5.0\% & 3.5\% & 4.3\% & 3.0\% & 25.7\% & 14.4\% \\ \hline
30 & 4.7\% & 2.1\% & 4.1\%  & 1.9\% & 48.3\% & 42.7\% \\ \hline
40 & 4.1\% & 3.1\% & 3.7\% & 2.8\% & 51.2\% & 39.5\% \\ \hline
\end{tabular}
\caption{(MovieLens) Final gap $P^T$, time-avg. gap $\widehat{P}$, stationary gap $D$ (avg. over 20 instances), varying $J$ ($K=10, Q=50, T=5)$.}
\label{table:movieworlds1}
\end{table}

\begin{figure}[t]
\centering
\begin{subfigure}
  \centering
  \includegraphics[width=.45\linewidth]{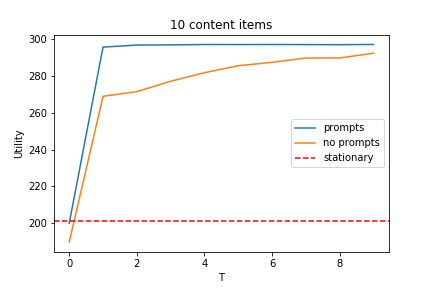}
\end{subfigure}
\begin{subfigure}
  \centering
  \includegraphics[width=0.45\linewidth]{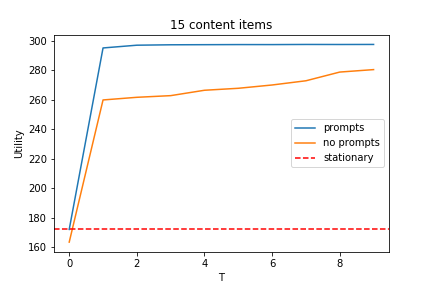}
\end{subfigure}
\begin{subfigure}
  \centering
  \includegraphics[width=.45\linewidth]{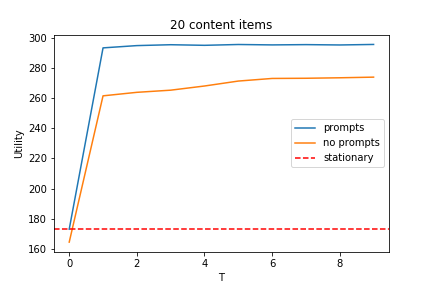}
\end{subfigure}
\begin{subfigure}
  \centering
  \includegraphics[width=.45\linewidth]{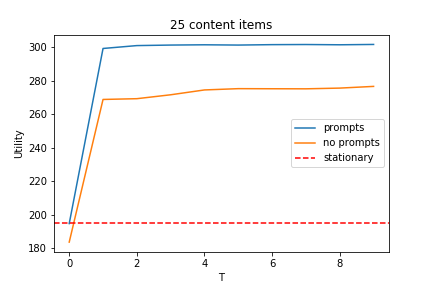}
\end{subfigure}
\caption{(Synthetic) Total per-period utility $E^t$ (avg.\ over 20 instances) for different policies, varying $J$ ($K = 5, T = 10, Q=30$).}
\label{fig:Q=30}
\end{figure}

\begin{figure}[t]
\centering
\begin{subfigure}
  \centering
  \includegraphics[width=.45\linewidth]{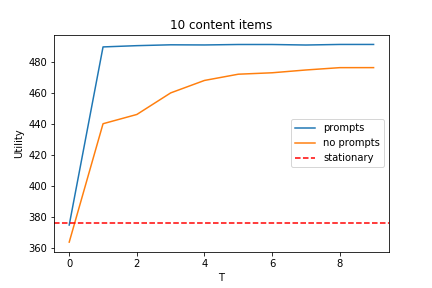}
\end{subfigure}
\begin{subfigure}
  \centering
  \includegraphics[width=0.45\linewidth]{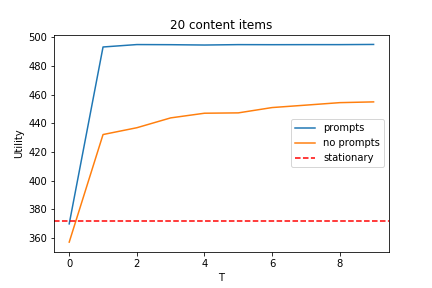}
\end{subfigure}
\caption{(Synthetic) Total per-period utility $E^t$ (avg.\ over 20 instances) for different policies, varying $J$ ($K = 5, T = 10, Q=50$).}
\label{fig:Q=50}
\end{figure}

\begin{figure}[t]
\centering
\begin{subfigure}
  \centering
  \includegraphics[width=.45\linewidth]{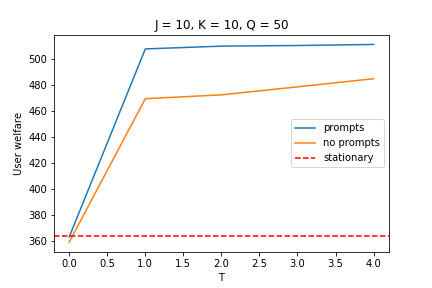}
\end{subfigure}
\begin{subfigure}
  \centering
  \includegraphics[width=0.45\linewidth]{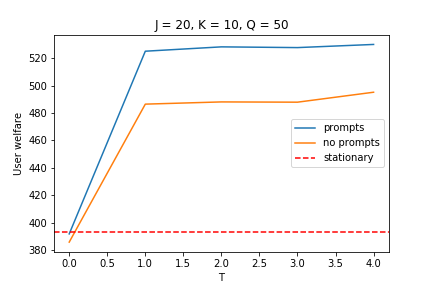}
\end{subfigure}
\begin{subfigure}
  \centering
  \includegraphics[width=.45\linewidth]{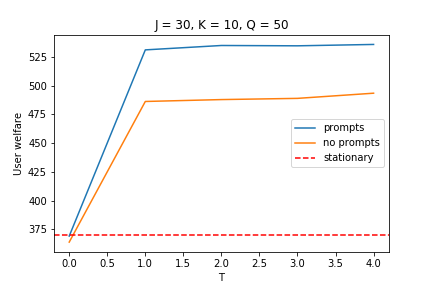}
\end{subfigure}
\begin{subfigure}
  \centering
  \includegraphics[width=.45\linewidth]{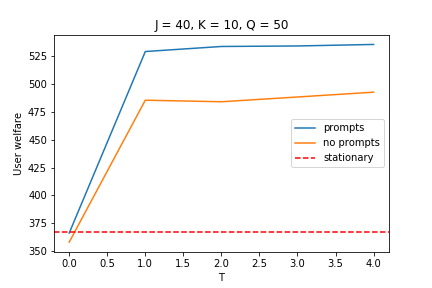}
\end{subfigure}
\caption{(Synthetic) Total per-period utility $E^t$ (avg.\ over 20 instances) for different policies, varying $J$ ($K = 10,Q = 50, T=5$).}
\label{fig:bigworlds}
\end{figure}

\begin{figure}[t]
\centering
\begin{subfigure}
  \centering
  \includegraphics[width=.45\linewidth]{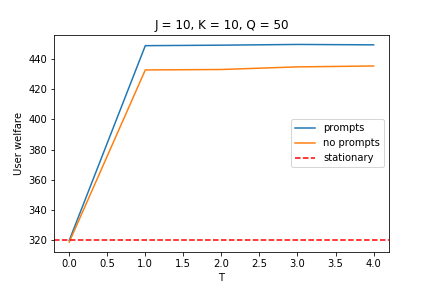}
\end{subfigure}
\begin{subfigure}
  \centering
  \includegraphics[width=0.45\linewidth]{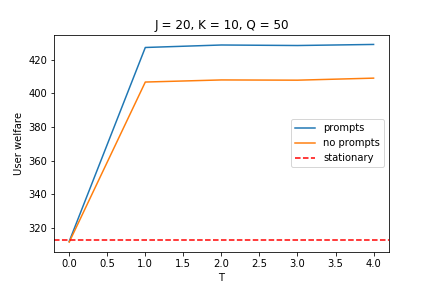}
\end{subfigure}
\begin{subfigure}
  \centering
  \includegraphics[width=.45\linewidth]{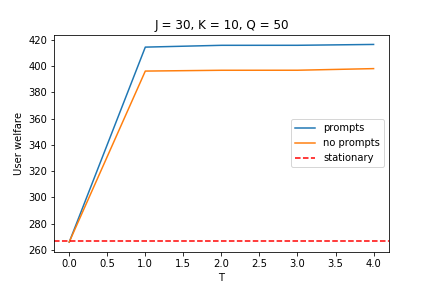}
\end{subfigure}
\begin{subfigure}
  \centering
  \includegraphics[width=.45\linewidth]{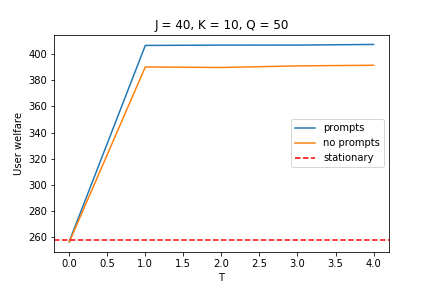}
\end{subfigure}
\caption{(MovieLens) Total per-period utility $E^t$ (avg.\ over 20 instances) for different policies, varying $J$ ($K = 10, Q = 50, T = 5$).}
\label{fig:movieworlds}
\end{figure}

\begin{figure}[t]
\centering
\begin{subfigure}
  \centering
  \includegraphics[width=.45\linewidth]{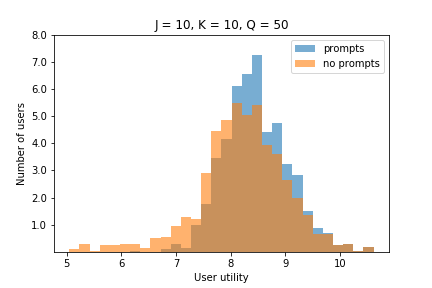}
\end{subfigure}
\begin{subfigure}
  \centering
  \includegraphics[width=0.45\linewidth]{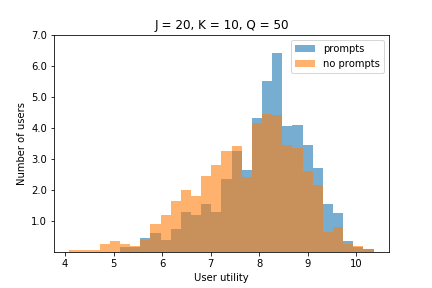}
\end{subfigure}
\begin{subfigure}
  \centering
  \includegraphics[width=.45\linewidth]{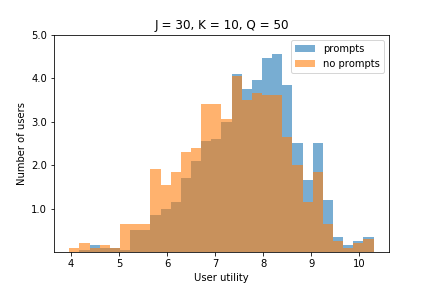}
\end{subfigure}
\begin{subfigure}
  \centering
  \includegraphics[width=.45\linewidth]{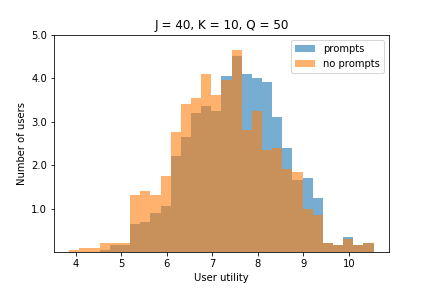}
\end{subfigure}
\caption{(MovieLens) User utility histograms (avg.\ over 20 instances), varying $J$ ($K = 10, Q = 50, T = 5$).}
\label{fig:movieworlds_utils}
\end{figure}

\begin{figure}[t]
\centering
\begin{subfigure}
  \centering
  \includegraphics[width=.45\linewidth]{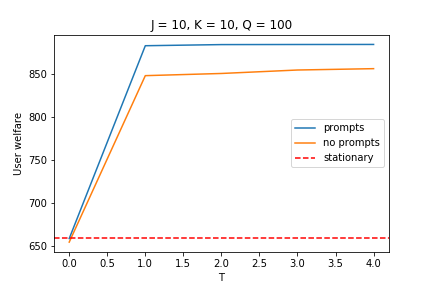}
\end{subfigure}
\begin{subfigure}
  \centering
  \includegraphics[width=0.45\linewidth]{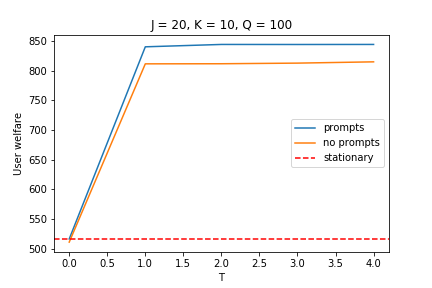}
\end{subfigure}
\caption{(MovieLens) Total per-period utility $E^t$ (avg.\ over 10 instances) for different policies, varying $J$ ($K = 10, Q = 100, T=5$).}
\label{fig:movieworlds2}
\end{figure}

\begin{figure}[t]
\centering
\begin{subfigure}
  \centering
  \includegraphics[width=.45\linewidth]{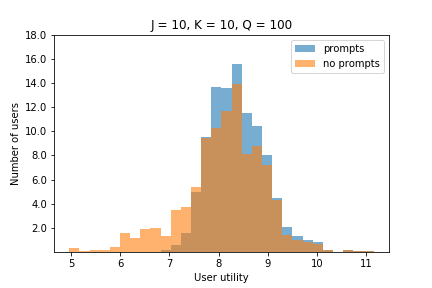}
\end{subfigure}
\begin{subfigure}
  \centering
  \includegraphics[width=0.45\linewidth]{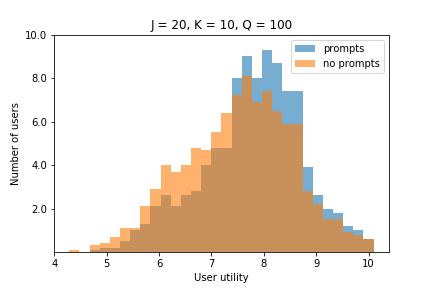}
\end{subfigure}
\caption{(MovieLens) User utility histograms (avg.\ over 10 instances), varying $J$ ($K = 10, Q = 100, T = 5$).}
\label{fig:movieworlds_utils2}
\end{figure}

\subsection{Column Generation Formulation}
\label{sec:app_cg}
Our column generation (CG) formulation differs considerably from the first (which we refer to as the {\em compact} formulation). Instead of articulating a (relatively) small integer program, it provides a large linear program whose solutions are likely to be integer points.  The idea is to enumerate (somewhat implicitly) the possible content location paths a provider can take. More precisely, for each provider and stage, we consider explicitly the audience (subset of users) matched to the provider at that stage, as well as any audience commitment made by an RS prompt. Given any \emph{entire sequence} of such matchings/prompts, assuming only that provider belief updates are deterministic, the location path of the provider is then fixed. This extended model enumerates all such sequences, requiring a doubly exponential number of variables, in contrast to the compact MIP formulation.
For this reason we develop a CG formulation. With CG, our second model has a number of advantages: the MIP is ``naturally linear,'' without the need to linearize quadratic terms; it also contains fewer, less complicated constraints; finally, it supports arbitrarily complex belief dynamics (requiring only deterministic update).

Let a \emph{star} be a subset $Q\subseteq \calQ$ of users, representing a possible audience that can be matched to a provider at a given stage. A \emph{multi-star} is a sequence of stars $\Bar{Q} = (Q_1, \ldots, Q_T)$ of length $T$ (one per stage). There are $2^{|\calQ|^T}$ multi-stars. A \emph{promise} $C \in \mathbb{R}^{J}_{\geq 0}$ is a non-negative vector with at most one non-zero coordinate (bounded by $|\calQ|$), reflecting the audience utility promised to a provider via some prompt, with a \emph{multi-promise} $\Bar{C} = (C_1, \ldots, C_T)$ defined in the obvious way. A \emph{treatment} $\tau = (\Bar{Q}, \Bar{C})$ reflects the impact of an RS matching/prompting policy on a provider; $\tau$ is \emph{honest} if the realised audiences at each stage meet (or exceed) the promises. Let $\calT$ be set of honest treatments.

Key to our MIP is the use of binary variables for each provider-treatment pair $\mu_{k,\tau}$, $k\in\calK, \tau \in\calT$ dictating the treatment (audiences and prompted commitments) associated with each provider. We omit full details,
but describe the main components and concepts underlying the formulation. First, we assume deterministic skill and audience belief updates (e.g., as described above). This means that, given the initial state $S^0$, and the treatment $\tau$ assigned to provider $k$, we can \emph{precompute} the exact utility $E_{k,\tau}$ across all $T$ stages, i.e., these are constants. With standard constraints (e.g., one treatment per provider; no user matched to more than one provider per stage, i.e., no ``star overlap''), we optimize total user utility: $\sum_{k\in\calK, \tau\in\calT} E_{k,\tau}\mu_{k,\tau}$. Notice that this approach makes very few assumptions about provider belief dynamics or best response behavior---it is all encoded in the computation of objective coefficients $E_{k,\tau}$.\footnote{Suppose that we allowed content points to fall anywhere in $\R^d$. Once a
star (or audience $Q\subseteq\calQ$) is fixed at stage $t$, the fact that $k$'s best response is to maximize utility dictates a unique content point suitable for that audience. Given the finite number of user subsets, this justifies the assumption of a finite (though possible large) set of viable content points. The size of this set, fortunately, plays no role when we use CG.}

The price one pays in this MIP is the explosion in the number of decision variables, which enumerates \emph{all treatments} (per provider). Of course, only a small number of these will be active at the optimal solution (or in \emph{any} feasible solution), so we can solve the problem using 
\emph{column generation (CG)} \cite{desrosiers-colgen:or2005}. In CG, we solve the problem using only a subset of these variables (or columns in the LP relaxation of the MIP), solve a \emph{subproblem} to determine the missing variable (column) whose addition offers the greatest marginal increase in objective value, then add that variable/column and re-solve. This process is repeated until no improving values are found or some other termination criterion is met (e.g., reaching a max number of iterations, or minimum improvement threshold). The reduced LP has one constraint for every provider $k$ and one for every user-stage pair $(q,t)$. The CG subproblem requires finding the variable with maximum \emph{reduced cost}: given the values of the dual variables (one per constraint, $\lambda_k$ for $k\in K$ and $\lambda_{qt}$ for $(q,t) \in \calQ \times [1, \ldots, T]$) at the optimal solution of the current reduced LP, we find the provider/treatment pair 
$$(k^\ast,\tau^\ast) = \argmax_{k \in K, \tau \in\calT} E_{k,\tau}  - \lambda_k - \sum_t \sum_{q \in Q_t^\tau}\lambda_{qt}$$ and add the corresponding variable $\mu_{k^\ast,\tau^\ast}$ to the LP. (Here $Q_t^\tau$ is the $t$th star in treatment $\tau$.)


Another advantage of our star/treatment formulation is evident: the dual decomposition reduces the subproblem to a collection of \emph{single-provider} location-path planning problems. We find the max reduced-cost improving treatment for each provider, then add the treatment for the maximizing provider.


\end{document}